\documentclass[a4paper,twocolumn,11pt]{quantumarticle}
\makeatletter
\providecommand\@afterenddocumenthook{}
\makeatother
\pdfoutput=1
\usepackage[utf8]{inputenc}
\usepackage{csquotes}
\usepackage[english]{babel}
\usepackage[T1]{fontenc}
\usepackage{amsmath}
\usepackage{amssymb}
\usepackage{hyperref}
\usepackage{braket}
\usepackage{bm}
\usepackage{comment}
\usepackage{float}
\usepackage{amsthm}
\usepackage[numbers,sort&compress]{natbib}

\newtheorem{theorem}{Theorem}
\newtheorem{corollary}{Corollary}

\makeatletter
\newrobustcmd{\fixappendix}{%
  \patchcmd{\l@section}{1.5em}{7em}{}{}%
  \patchcmd{\l@subsection}{2.3em}{7em}{}{}%
}
\makeatother

\newcommand{\eg}{e.g$.$ }

\newcommand{\ie}{i.e$.$ }
\newcommand{\cf}{cf$.$ }

\newcommand{\wrt}{wrt$.$ }
\newcommand{\dd}{\ensuremath{\mathrm{d}}}

\newcommand{\secref}[1]{Section~\ref{sec:#1}}
\newcommand{\appref}[1]{Appendix~\ref{app:#1}}
\newcommand{\figref}[1]{Figure~\ref{fi:#1}}
\newcommand{\eqnref}[1]{Equation~\ref{eq:#1}}

\begin{document}

\title[Wave packets from the spectrum]{Wave packets from the spectrum}

\author[2,3]{ChunJun Cao}
\author[4,5]{Oliver Friedrich}
\email[corresponding, ]{oliver.friedrich@lmu.de}
\author[2,3]{Marin Girard}
\author[6]{Nicolas Loizeau}
\author[7]{Ashmeet Singh}

\affil[2]{Virginia Tech Center for Quantum Information Science and Engineering, Blacksburg, VA 24061, USA}

\affil[3]{Department of Physics, Virginia Tech, Blacksburg, VA 24061, USA}
\affil[4]{University Observatory, Faculty of Physics, Ludwig-Maximilians-Universität, Scheinerstr.\ 1, 81677 Munich, Germany\textbackslash EU}
\affil[5]{Excellence Cluster ORIGINS, Boltzmannstr.\ 2, 85748 Garching, Germany\textbackslash EU}
\affil[6]{Niels Bohr Institute, University of Copenhagen, Copenhagen, Denmark\textbackslash EU}
\affil[7]{Department of Physics, Indian Institute of Technology Delhi, Hauz Khas, New Delhi, 110016, India}

\begin{abstract}
    The freedom to change Fock basis seems to ensure a minimum amount of locality in lattice theories in the following sense: If $\lbrace (\hat a_i^\dagger\,,\,\hat a_i)\rbrace$ for $i=1,\dots,n$ is a lattice of creation and annihilation operators and if a given Hamiltonian $\hat H$ induces highly non-local dynamics on that lattice, then it will usually be possible to change to a new set of operators $\lbrace (\hat b_i^\dagger\,,\,\hat b_i)\rbrace$ in terms of which the dynamics appear less non-local. We demonstrate this by turning a highly non-local random matrix model  into a local, 1D lattice theory where particles can propagate in localized wave packets. More generally, we show that any Hamiltonian can be made to look like such a theory, with the lattice dispersion relation and the non-integrability of the theory depending on the spectrum of $\hat H$. We argue that our results are a step towards quantum mereology for fields.
\end{abstract}

\newpage 

\tableofcontents

\begin{figure*}
\centering  \includegraphics[width=\textwidth]{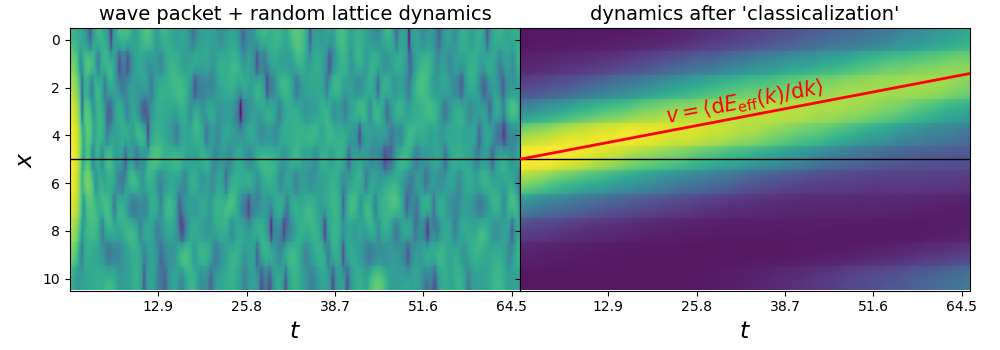}
   \caption{Demonstrating our procedure for (re-)interpreting a given Hamiltonian $\hat H$ as a lattice theory with local interaction structure. Left panel: evolution of an initially Gaussian wave packet \wrt random lattice dynamics, \ie when choosing the coefficients $h^\ell$ in \eqnref{decomp_intro} at random. Right panel: switching from the initial set of lattice operators $\hat a_i^\dagger\,,\,\hat a_i$ to a new set $\hat b_i^\dagger\,,\,\hat b_i$ via the algorithm developed in this paper. Wrt.\ the new lattice, states can stay localized for some time \& propagate according to an effective dispersion relation. We considered a Hilbert space of dimension $2^{11}$ (corresponding to 11 qubits) for this figure - see Sections~\ref{sec:intro} and \ref{sec:wave_packets} for details.}
   \label{fi:Rogerio_packet}
\end{figure*}

\section{Classicalizing quantum theories: quantum mereology}
\label{sec:intro}


Consider a theory of $n$ lattice sites with creation and annihilation operators at each site given by $\hat a_i^\dagger\,,\,\hat a_i\,$. For concreteness, we will consider a Fermion lattice such that these operators satisfy the anti-commutation relation
\begin{equation}
    \lbrace \hat a_i^\dagger\,,\,\hat a_j\rbrace = \delta_{ij}\ .
\end{equation}
At each lattice site we can define the Pauli-like operators
$\hat X_i \equiv\ (\hat a_i^\dagger + \hat a_i)$ , $\hat Y_i \equiv\ \mathrm{i}(\hat a_i^\dagger - \hat a_i)$ and $\hat Z_i \equiv 2 a_i^\dagger a_i - 1$,
and from these we can build (Fermionic)
\emph{Pauli strings} $\hat \tau$ as
\begin{align}
\hat \tau = \hat O_1\, \hat O_2\, \dots \, \hat O_n \ \ \mathrm{where}\ \  \hat O_i \in \lbrace \mathbb{I}\,,\,\hat X_i\,,\,\hat Y_i\,,\,\hat Z_i\rbrace\ .\nonumber \\
\end{align}
The set of all such strings forms a basis for operators on the lattice, and \eg any Hamilton operator $\hat H$ can be expressed as
\begin{equation}
\label{eq:decomp_intro}
    \hat H \simeq \sum_\ell h^\ell\,\hat\tau_\ell\ ,
\end{equation}
where $\hat\tau_\ell$ is the $\ell$th element in a list of all possible Pauli strings.

How should we choose the coefficients $h^\ell$ in order to create highly non-local dynamics on the above lattice? One attempt may be to randomly draw coefficients of similar amplitude for \emph{all} Pauli strings. Terms that couple arbitrarily many sites that are arbitrarily far from each other would then be present in \eqnref{decomp_intro}\ . But even if we draw such a Hamiltonian, there may still be a different set of lattice operators $\hat b_i^\dagger\,,\,\hat b_i$ \wrt which the dynamics of $\hat H$ take a less non-local form than \wrt the $\hat a_i^\dagger\,,\,\hat a_i\,$. So \wrt to this new lattice the dynamics of $\hat H$ would actually appear less extreme than what we had intended it to be.


We demonstrate such a situation in \figref{Rogerio_packet}. The left panel there shows the dynamics of an initially Gaussian wave packet when the coefficients $h^\ell$ are drawn from independent \& identical Gaussian distributions (+ an additional factor of $\mathrm{i}$ for Pauli strings with an even number of factors $\hat X_i$ or $\hat Y_i$ to ensure that $\hat H$ is Hermitian; \cf Appendix~\ref{app:GUE} for details). The right panel then switches from the defining lattice operators $\hat a_i^\dagger\,,\,\hat a_i$ to a new set of operators $\hat b_i^\dagger\,,\,\hat b_i$ - via an algorithm that we develop in this paper. The color map in both panels  indicates the value of $|\langle i|\psi\rangle|^2$, where $\ket{i}$ is given by either $\hat a_i^\dagger \ket{0}$ or by $\hat b_i^\dagger \ket{0}$ and where $\ket{\psi}$ is a Gaussian wave packet \wrt either the operators $\hat a_i^\dagger\,,\,\hat a_i$ or the operators $\hat b_i^\dagger\,,\,\hat b_i\,$. Note that the random dynamics do not preserve particle number, such that probability leaks from the one-particle wave function $\braket{i|\psi}$ to subspaces of higher particle number. For that reason, the color scale in the left panel of \figref{Rogerio_packet} was chosen to be logarithmic in $|\langle i|\psi\rangle|^2$. In \secref{wave_packets} we find that as a general property of our algorithm particle number is preserved \wrt the new lattice operators $\hat b_i^\dagger\,,\,\hat b_i\,$ (hence, a linear color scale was chosen in the right panel). It is also apparent from the right panel of \figref{Rogerio_packet} that \wrt the new lattice, initially localized states can stay localized for some time and propagate according to some effective dispersion relation - see again \secref{wave_packets} for details.

In general, one might expect that for any Hamiltonian $\hat H$ in a Hilbert space of dimension $d = 2^n$ there will be a lattice of operators $\hat b_i^\dagger\,,\,\hat b_i\,$ that makes the dynamics induced by $\hat H$ look most local. In this paper we explore how to find such operator sets. Our investigation builds on the point-of-view that the physics of our Universe is solely described by unitary dynamics of states in a Hilbert space $\mathcal{H}$ induced by a Hamiltonian $\hat H$. Any decomposition of $\mathcal{H}$ into tensor factors or lattice sites is not fundamental, but can be chosen to satisfy certain operational purposes.

This point-of-view puts our work within the program of \emph{quantum mereology} \cite{Carroll_Singh2020}, which intends to work out a map from the landscape of abstract quantum theories (i.e.\ from quantum states and/or Hamiltonian operators that are not explicitly derived by \enquote{quantizing} classical theories) to degrees-of-freedom with semi-classical behaviour. A high-level summary of (one version of) the mereology program could be given as follows: For a given Hamiltonian $\hat{H}$, split the Hilbert space $\mathcal{H}$ into tensor factors
\begin{equation}
    \mathcal{H} = \mathcal{H}_1 \otimes \mathcal{H}_2 \otimes \mathcal{H}_3 \otimes\ \dots
\end{equation}
such that on each factor $\mathcal{H}_i$ pairs of conjugate observables $\lbrace\hat{Q}_i,\hat{P}_i\rbrace$ can be defined whose dynamics -- as induced by the Hamiltonian $\hat{H}$ -- is close-to-classical. We would then think of these observables as measuring the different degrees-of-freedom (dofs) of the emergent, semi-classical Universe.



Many ways of decomposing an abstract quantum theory into degrees-of-freedom and quantifying the classicality of such decompositions have been explored. The authors of \cite{Carroll_Singh2020}, motivated by the quantum decoherence paradigm, split a Hilbert space into a single system and an environment, i.e.\ $\mathcal{H} = \mathcal{H}_{\mathrm{sys}}\otimes \mathcal{H}_{\mathrm{env}}\,$, and identified conjugate operator pairs on $\mathcal{H}_{\mathrm{sys}}$ such that two criteria were optimized for: entanglement growth between system and environment should be minimized and peaked states in the emergent conjugate operators (\enquote{wave packets}) should disperse slowly. The authors referred to these two conditions as \emph{robustness} and \emph{predictability} of the emergent classical factorization.

In \cite{Cotler2019} the authors investigated whether a given Hamiltonian $\hat{H}$ admits a factorization of $\mathcal{H}$ into a set of factors for which $\hat{H}$ encodes \emph{local} interactions. In particular, they were looking for $k$-local factorizations where the interaction terms appearing in $\hat{H}$ connect at most $k$ of the emergent Hilbert space factors. Their findings are twofold: if a $k$-local factorization for a given Hamiltonian exists (for $k \ll N_{\rm{dof}}$), then it is essentially unique \cite{soulas2025a}; but most Hamiltonians do not admit an exact $k$-local factorization. It was however demonstrated by \cite{Loizeau2023} that random Hamiltonians drawn from the Gaussian Unitary Ensemble are well approximated even by $2$-local Hamiltonians in some factorization. In particular, they showed that the error in the spectrum made by this approximation decreases exponentially with system size.

Of course, factorizing a Hilbert space such that the Hamiltonian's interaction terms are close to $2$-local on the graph of factors is not yet sufficient for representing genuinely local interactions of degrees-of-freedom that are arranged in some low-dimensional geometry. The authors of \cite{Cao2017, CaoCarroll2018} investigated how such a geometry can be assigned to the graph of Hilbert space factors, using quantum entanglement as a proxy for geometrical concepts such as \enquote{distance} and \enquote{area}. Since entanglement is a property of the state (and not of the Hamiltonian), this suggests that mereology should be informed by both $\hat H$ and the current state of a quantum system, in order to identify the Hilbert space factorization that most closely resembles a classical Universe. In \cite{Pollack_Singh2019}, the authors argued that the vast majority of finite-dimensional Hilbert spaces cannot be isomorphic to the tensor product of Hilbert space sub-factors that describes a lattice theory. They explored emergent direct-sum decompositions/superselection sectors between which scrambling of localized information is minimized. Each sector, with the appropriate dimensions, can then be explored for its suitability to serve as a lattice-like framework.

A generalization of the mereology framework was proposed by \cite{Zanardi2004, Zanardi2024,zanardi2025}, who suggested that degrees-of-freedom need not be represented in terms of exact Hilbert space factors. Instead, any von-Neumann algebra of operators $\mathcal{A}$ could represent the operational degrees-of-freedom of a (sub-)system of the Universe. By minimizing the \emph{Gaussian scrambling rate} between $\mathcal{A}$ and its commutant $\mathcal{A}'$, the authors of \cite{Zanardi2024} identified the most independent algebras of a given theory. This approach can still be thought of as splitting a quantum theory into a system and its environment, with the objective of minimizing the exchange of quantum information between the two. In the case where $\mathcal{A}$ is a factor algebra, \cite{Zanardi2024} also showed that minimization of the Gaussian scrambling time is equivalent to finding the system-environment split in which the interaction Hamiltonian of the two factors has minimized Hilbert-Schmidt norm. The authors of \cite{Kabernik2020} also studied generalizations of exact tensor products. In particular, they considered \emph{partial bipartitions}, structures that are weaker than algebraic factors, to encode information about coarse-grained and collective degrees-of-freedom which can be thought of as classical under Hamiltonian evolution. 

Instead of splitting a given Hilbert space $\mathcal{H}$ into a single system and an environment, \cite{Loizeau2024} proposed to iteratively decompose $\mathcal{H}$ into a more and more fine grained set of \emph{weakly interacting sectors}. They would e.g.\ first decompose a Hilbert space of dimension $2^N$ as $\mathcal{H} = \mathcal{H}_A \otimes \mathcal{H}_B$ where the two factors have dimension $2^{N/2}$ and are chosen such that in this factorization the Hamiltonian of a given theory is well approximated by $\hat H \approx \hat H_A \otimes \mathbb{I}_B + \mathbb{I}_A \otimes \hat H_B$\ . The resulting factors $\mathcal{H}_A$ and $\mathcal{H}_B$ can then be split further into an emergent hierarchy of weakly interacting sectors, and the spectrum of $\hat{H}$ determines how far this decomposition can go while still returning factors with weak interactions. The most fine grained factors that are still weakly interacting would then represent the elementary degrees-of-freedom of a given theory.

The authors of \cite{Adil2024} however point out that classicality of different sub-systems is often associated with the existence of long-lived pointer states, and not necessarily with weak interactions. They devise an algorithm that splits $\mathcal{H}$ into a system and an environment, such that a given Hamiltonian $\hat H$ permits the existence of such pointer states in the system. They consider different examples of random Hamiltonian operators and find that a given Hamiltonian often allows for a variety of different splits $\mathcal{H} = \mathcal{H}_{\mathrm{sys}} \otimes \mathcal{H}_{\mathrm{env}}$ such that $\mathcal{H}_{\mathrm{sys}}$ has stable pointer states. Interestingly, the same Hamiltonian can appear in very different shapes in the different emergent factorizations. Hence, \cite{Adil2024} call these factorizations the different \emph{realms} of emergent physics.

Finally, note that quantum mereology can also be interpreted as a procedure to extract symmetries from a generic quantum system.
Closed, strongly interacting many-body systems are expected to thermalize, in the sense captured by the Eigenstate Thermalization Hypothesis (ETH) \cite{Mori2018, Polkovnikov2011, Rigol2008}. But the semi-classical world of our experience does not show full eigenstate thermalization: empirically, not all local observables thermalize, and long-lived structures exist. This departure from ergodicity is commonly attributed to the presence of symmetries, both in the form of conserved quantities and dynamical symmetries \cite{buca2023,Sala2020, loizeau2025}.
Why is the world governed by 
such a non-generic Hamiltonian? Is its structure finely tuned to introduce symmetries, or can symmetries emerge within more generic systems - e.g.\ modeled by a random matrix Hamiltonian?
Random matrix theory successfully captures the spectral properties of strongly interacting quantum systems \cite{Wigner2023, Guhr1998, Atas2013, AlbrechtIglesias2008, AlbrechtIglesias2015}, but explaining the emergence of macroscopic, quasi-classical subsystems and associated symmetries requires going beyond spectral statistics to quantum mereology. 
The decomposition of a Hamiltonian into approximately non-interacting sectors, e.g.\ as $\hat H \approx \hat H_A \otimes \mathbb{I}_B + \mathbb{I}_A \otimes \hat H_B$ --- singles out conserved quantities: $\hat H_A$ commutes with $H$ and is therefore conserved under the dynamics. More generally, any observable supported on subsystem $A$ would fail to thermalize (quickly) with respect to the whole system, giving rise to a family of local conserved quantities and dynamical symmetries.
Other algebraic approaches to partitioning that take into account more general symmetries have been explored in \cite{vanrietvelde2025}.

A (rough) categorization of the different approaches to quantum mereology may be based on the different types of degrees-of-freedom they try to identify. We attempt to call two major categories as follows.

\begin{itemize}
    \item[] \textbf{A) Object mereology:}
    
    These approaches attempt to identify dofs associated with some semi-classical object, e.g.\ the position and momentum of a macroscopic body. The existence of such an object is usually associated with the existence of pointer states in the corresponding Hilbert space factors. Thus, object mereology will generally try to identify factors (or more generally, degrees-of-freedom) that permit such states (see e.g.\ \cite{Carroll_Singh2020, Adil2024}). 
    \item[]  \textbf{B) Field mereology:}

    Mereology may also attempt to identify meaningful dynamical building blocks of a theory which, as a whole, has a classical limit. In particular, one would not consider the conjugate Fourier mode pairs $\lbrace\hat\phi_{\bm{k}}\, ,\, \hat\pi_{\bm{k}}\rbrace$ in a quantum field theory as the degrees-of-freedom of any semi-classical object. But it would nevertheless be considered a successful implementation of mereology if we could identify Hilbert space factors and associated operator pairs that behave like field theory modes in an abstract quantum theory.
\end{itemize}
Our work here belongs to the 2nd category. In particular, we try to add to the mereology program by studying how an arbitrary Hamiltonian may be expressed in terms of low-interaction degrees-of-freedom and we argue that these may serve as the harmonic modes of a lattice theory.

In \secref{outlook} we describe an algorithm for iteratively decomposing any Hamiltonian of dimension $d=2^n$ into a set of $n$ qubits that can be interpreted as modes of such a primitive lattice theory. And we also present results regarding the structure of any given Hamiltonian \wrt minimum-interaction-factorizations. Section~\ref{sec:app_to_GUE} in particular demonstrates that we can turn a chaotic random matrix model into a theory of propagating particles, whose propagation speed is described by an effective dispersion relation - \cf Figure~\ref{fi:Wigner_packets} for an advance look at these results. In \secref{method} we shed light on some of the findings of \secref{outlook} by showing that GUE Hamiltonians admit a discrete set of qubits with particularly low interaction. \secref{discussion} concludes with a discussion. In particular, we describe there open tasks that need to be addressed on a path towards retrieving field theory via quantum mereology.

\section{From abstract \texorpdfstring{$\hat H$}{} to lattice theory with local interactions}
\label{sec:outlook}
\label{sec:wave_packets}

Given a quantum theory defined by some Hamiltonian $\hat H$ acting on some Hilbert space $\mathcal{H}$ we want to identify a lattice of operators such that $\hat H$ takes the form of a local lattice theory \wrt these operators. For simplicity, we will be looking for a lattice of spin-$\frac{1}{2}$ particles which we will also refer to as a lattice of (Fermionic) qubits. Let $\mathcal{H}$ e.g.\ be of dimension $d=2^n$. Such a space is isomorphic to the Fock space of $n$ such qubits, 
\begin{align}
\label{eq:isomorphism}
    &\mathcal{H} \simeq \nonumber \\
    &\mathrm{span}\left(\left.\ket{i_1\dots i_n}\equiv(\hat a_1^\dagger)^{i_1}\dots (\hat a_n^\dagger)^{i_n}\ket{0}\right|i_k\in\lbrace 0,1\rbrace\right)
\end{align}
where the operators $\hat a_i^\dagger$ satisfy the anti-com-mutation relations
\begin{equation}
    \lbrace \hat a_i, \hat a_j^\dagger\rbrace = \delta_{ij}
\end{equation}
and where $\ket{0}$ is a state that is annihilated by all $\hat a_i\,$. 

Once we have identified some basis of $\mathcal{H}$ with the Fock basis $\lbrace \ket{i_1\dots i_n} \rbrace$ we can define Pauli-like operators $\hat X_i \equiv\ (\hat a_i^\dagger + \hat a_i)$ , $\hat Y_i \equiv\ \mathrm{i}(\hat a_i^\dagger - \hat a_i)$ and $\hat Z_i \equiv 2 a_i^\dagger a_i - 1$ at each lattice site. From these we can build operators of the form
\begin{equation}
\hat O_1\, \hat O_2\, \dots \, \hat O_n \ \ \mathrm{where}\ \  \hat O_i \in \lbrace \mathbb{I}\,,\,\hat X_i\,,\,\hat Y_i\,,\,\hat Z_i\rbrace\ .
\end{equation}
In \secref{intro} we called these the (Fermionic) \emph{Pauli strings}. In the following we want to slightly modify this convention. For $j\neq i$ the different $\hat X_i, \hat Y_i$ and $\hat X_j, \hat Y_j$ are anti-commuting. This has the consequence that the above products are Hermitian if they contain an overall odd number of $\hat X_i$ and/or $\hat Y_i$, while they are anti-Hermitian if that number is even. In order to turn all Pauli strings into Hermitian operators we will thus multiply a factor $\mathrm{i}$ to those strings that contain an even number of $\hat X_i$ and/or $\hat Y_i\,$. With this convention, the Pauli strings become a complete basis for Hermitian operators on $\mathcal{H}$ (viewing the latter as a real vector space). Any Hamiltonian operator $\hat H$ can thus be expressed as
\begin{equation}
\label{eq:expansion}
    \hat H \simeq \sum_\ell h^\ell\,\hat\tau_\ell\ ,
\end{equation}
where $\hat\tau_\ell$ is the $\ell$th element in a list of all Pauli strings and $h^\ell$ are real coefficients. 

Any choice of Fock basis or, equivalently, any choice of lattice operators $\hat a_i\,,\,\hat a_i^\dagger$ already provides an interpretation of $\hat H$ as a "lattice theory" via \eqnref{expansion}. But by choosing the operators $\hat a_i\,,\,\hat a_i^\dagger$ in a clever way one may hope that the expansion in \eqnref{expansion} can be brought closer to a local interaction structure.

Our starting point for finding such a set of operators is the realization that integrable lattice theories can always be decomposed into a set of modes that have vanishing interaction between each other - these factors being the "Fourier modes" of the lattice field. So we will try to choose the lattice operators $\hat a_i\,,\,\hat a_i^\dagger$ such that the interaction between the lattice sites is minimized. These operators can serve as the Fourier modes of an emergent lattice field (whose dynamics is approximately integrable).

In \secref{structure_of_H} we introduce a procedure for splitting off individual low-interaction qubits from a given $\mathcal{H}$ and $\hat H$. In \secref{iterative_algorithm} we show how to iteratively apply this procedure to obtain a lattice-like theory from any Hamiltonian. We then apply this iterative scheme to the Hamiltonian of a Pauli field (\secref{app_to_Pauli}) and to a Hamiltonian drawn from the Gaussian Unitary Ensemble (GUE, \secref{app_to_GUE}). Surprisingly, our approach turns the GUE Hamiltonian into a theory that allows for wave packets that propagate and stay local on the emergent lattice. 

\subsection{Candidate dofs from minimum interaction factorizations}
\label{sec:structure_of_H}

Let us start by considering a Hilbert space of even dimension $d$ that factorizes as $\mathcal{H} = \mathcal{H}_q \otimes \mathcal{H}_r$, where $\mathcal{H}_q$ is a two-dimensional Hilbert space (that of a qubit). Given such a factorization a Hamiltonian $\hat H$ can always be written as
\begin{align}
    \hat H\, &\simeq\,  \hat H_{q} \otimes \mathbb{I}_r\, +\, \mathbb{I}_q \otimes \hat H_{r}\, +\, \frac{\mathrm{Tr}(\hat H)}{d}\cdot \mathbb{I}\, +\, \hat H_{\mathrm{int}}\ ,
\end{align}
where the self-Hamiltonians $\hat H_{q}$ and $\hat H_{r}$ are traceless and where already the partial traces of the interaction $\hat H_{\mathrm{int}}$ vanish, i.e.\ $\mathrm{Tr}_q(\hat H_{\mathrm{int}}) = 0$ and $\mathrm{Tr}_r(\hat H_{\mathrm{int}}) = 0\,$. Because the trace of $\hat H_{q}$ vanishes, we can consider it as a multiple of a Pauli matrix. Let us e.g.\ always choose a basis in $\mathcal{H}_q$ such that this Pauli matrix is $\sigma_z$. Then
\begin{equation}
    \hat H_{q}\, =\, \begin{pmatrix}
    E_q/2 & 0 \cr 0 & -E_q/2
    \end{pmatrix}
    \, \equiv\, \frac{E_q}{2}\, \sigma_z\ ,
\end{equation}
where $E_q$ is the (self-)energy required to excite the qubit in $\mathcal{H}_q$ from one level to the other. Also, since we can always subtract the global trace from $\hat H$ without changing the dynamics it induces, we will from now on assume that $\mathrm{Tr} \hat H = 0\,$.  The Hamiltonian then becomes
\begin{align}
\label{eq:Hamiltonian_with_Sig_z}
    \hat H\, &\simeq\, \frac{E_q}{2}\, \sigma_z \otimes \mathbb{I}_r\, +\, \mathbb{I}_q \otimes \hat H_{r}\, +\, \hat H_{\mathrm{int}}\ .
\end{align}
For $d\geq 4$ there are infinitely many ways to factorize $\mathcal{H}$ as $\simeq \mathcal{H}_q \otimes \mathcal{H}_r\,$, and the different terms in the decomposition of the Hamiltonian in \eqnref{Hamiltonian_with_Sig_z} will in general differ between any two factorizations. In particular, the size of the interaction $\hat H_{\mathrm{int}}$ between the qubit $\mathcal{H}_q$ and the rest of the Hilbert space, $\mathcal{H}_r$, will be different for different factorizations. Here we are interested in factorizations where the interaction Hamiltonian is minimized, because operators that are supported only in $\mathcal{H}_q$ may then be interpreted as \enquote{harmonic modes} of an emergent, approximately integrable lattice theory. 

For concreteness, let us minimize the Hilbert-Schmidt norm of the interaction Hamiltonian, \ie we minimize the loss function
\begin{equation}
    \label{eq:loss_function}
    \mathcal{L}(\mathcal{H} \simeq \mathcal{H}_q \otimes \mathcal{H}_r|\hat H) \equiv \mathrm{Tr}(\hat H_{\mathrm{int}}^2)
\end{equation}
in the space of possible factorizations $\mathcal{H} \simeq \mathcal{H}_q \otimes \mathcal{H}_r\,$. In Appendix~\ref{app:structure_of_Hint} we demonstrate the following result.

\begin{theorem}
\label{th:H_structure}
    Wrt.\ factorizations $\mathcal{H} \simeq \mathcal{H}_q \otimes \mathcal{H}_r$ that are local minima of the loss function $\mathrm{Tr}(\hat H_{\mathrm{int}}^2)$ (in the space of possible factorizations) the (traceless) Hamiltonian $\hat H$ can always be expressed as
    \begin{align}
        &\hat H\, \simeq\, \frac{E_q}{2}\,\sigma_z \otimes \mathbb{I}_r\, +\, \mathbb{I}_q \otimes \hat H_{r}\, +\, \sigma_z \otimes \hat X_{\mathrm{int}}
    \end{align}
    i.e.\ the interaction Hamiltonian is always of the form $\hat H_{\mathrm{int}} \simeq \sigma_z \otimes \hat X_{\mathrm{int}}\,$. Furthermore the operator $\hat X_{\mathrm{int}}$ commutes with $\hat H_{r}\,$:
    \begin{equation}
        [\hat H_{r}, \hat X_{\mathrm{int}}] = 0\ .
    \end{equation}
\end{theorem}

\begin{corollary}
\label{corr:H_structure}
    For factorizations $\mathcal{H} \simeq \mathcal{H}_q \otimes \mathcal{H}_r$ that are local minima of the loss function $\mathrm{Tr}(\hat H_{\mathrm{int}}^2)$ the auto-Hamiltonians $\hat H_{q} \otimes \mathbb{I}_r$ and $\mathbb{I}_q \otimes \hat H_{r}$ and the interaction Hamiltonian $\hat H_{\mathrm{int}}$ can always be simultaneously diagonalized.
\end{corollary}

\noindent We will consider the qubits $\mathcal{H}_q$ of factorizations $\mathcal{H} \simeq \mathcal{H}_q \otimes \mathcal{H}_r\,$ that are local minima of $\mathcal{L}$ to represent \emph{candidate degrees-of-freedom}. This is motivated by the fact that any small rotation away from such a factorization can only increase the interaction between the factors.

\subsection{Iteratively retrieving a Fourier space lattice}
\label{sec:iterative_algorithm}

We now want to move to a complete decomposition of a Hilbert space $\mathcal{H}$ of dimension $d=2^n$ into a lattice of $n$ qubits. We will start by looking for a Bosonic lattice, \ie by decomposing $\mathcal{H}$ into a tensor product of $n$ qubit Hilbert spaces $\mathcal{H}_{q,i}\,$,
\begin{equation}
\label{eq:isomorphism_2}
    \mathcal{H} \simeq \bigotimes_{i=1}^n \mathcal{H}_{q,i}\ .
\end{equation}
Given a Hamiltonian operator $\hat H$ we will attempt to choose this factorization such that interaction between the different factors is small. Certain operators on these factors can then - after a further Boson-to-Fermion mapping - serve as Fourier modes of an approximately integrable lattice theory.

To find such a factorization we will iteratively factor out individual low-interaction qubits in the way we had described in the previous subsection. More concretely, we proceed as

\begin{itemize}
    \item[Step 1:] $\mathcal{H} \simeq \mathcal{H}_{q,1}\otimes \mathcal{H}_{r,1}$ such that interaction between $\mathcal{H}_{q,1}$ and $\mathcal{H}_{r,1}$ is minimized;
    \item[Step 2:]  $\mathcal{H}_{r,1} \simeq \mathcal{H}_{q,2}\otimes \mathcal{H}_{r,2}$ such that interaction between $\mathcal{H}_{q,2}$ and $\mathcal{H}_{r,2}$ is minimized;

    \dots
    \item[Step $n-1$:]  $\mathcal{H}_{r,n-2} \simeq \mathcal{H}_{q,n-1}\otimes \mathcal{H}_{q,n}$ such that interaction between $\mathcal{H}_{q,n-1}$ and $\mathcal{H}_{q,n}$ is minimized.
\end{itemize}

\noindent This procedure is ambiguous, because at each iteration there will be many local minima of $\mathrm{Tr}(\hat H_{\mathrm{int}}^2)\,$, \ie many candidate degrees-of-freedom to choose from. For definiteness, let us assume that at each iteration we choose the minimum for which the ratio $\mathrm{Tr}(\hat H_{\mathrm{int}}^2) \big/ \mathrm{Tr}(\hat H_{\mathrm{q}}^2)$ is smallest. This ensures that we pick degrees-of-freedom, which are robust (\textit{i.e.} have high self-energies).

Defining
\begin{equation}
\label{eq:def_of_Zi_2nd}
    \hat Z_i = \underbrace{\mathbb{I} \otimes \cdots \otimes \mathbb{I}}_{i-1}\;\otimes\;\sigma_z\;\otimes\;\underbrace{\mathbb{I} \otimes \cdots \otimes \mathbb{I}}_{n-i}\ ,
\end{equation}
and remembering our conventions around \eqnref{Hamiltonian_with_Sig_z}, the self-Hamiltonian part of the total Hamiltonian $\hat H$ will in this iterative factorization be given by
\begin{equation}
\label{eq:H_self}
    \hat H_{\mathrm{self}} \simeq \sum_{i=1}^n \frac{E_{i}}{2}\, \hat Z_i\ .
\end{equation}
This is indeed reminiscent of the Hamiltonian of integrable lattice models, and the self-energies $E_i$ would be the energies associated with the Fourier modes of such a lattice. At the same time Theorem~\ref{th:H_structure} and Corollary~\ref{corr:H_structure} apply at each iteration of our factorization procedure, thus ensuring that the total interaction $\hat H_{\mathrm{int}}$ commutes (and is simultaneously diagonalisable) with $\hat H_{\mathrm{self}}\,$. This leads to the following consequence.

\begin{corollary}
\label{corr:Hint_from_Zs}
    Consider a Hamiltonian $\hat H$ on a Hilbert space $\mathcal{H}$ of dimension $d=2^n$, and iteratively split $\mathcal{H}$ into $n$ minimum-interaction qubits such that 
    \begin{equation}
    \label{eq:from_H_to_field_2}
        \hat H \simeq \sum_{i=1}^n \frac{E_{q,i}}{2}\, \hat Z_i \ +\ \hat H_{\mathrm{int}}\ .
    \end{equation}
    Then the total interaction $\hat H_{\mathrm{int}}$ is a superposition of Pauli strings that contain exclusively either $\mathbb{I}_{q}$ or $\sigma_{z}$ on each qubit factor.
\end{corollary}
\begin{proof}
If $\hat H_{\mathrm{int}}$ had support on a Pauli-string that contained either $\sigma_x$ or $\sigma_y$ on one of the qubit factors, then it could not be simultaneously diagonalizable with  $\sum_{i=1}^n \frac{E_{q,i}}{2}\, \hat Z_i\,$.
\end{proof}

If our iterative way of factorizing $\mathcal{H}$ into low-interaction qubits has produced a total interaction $\hat H_{\mathrm{int}}$ that is small compared to the self-Hamiltonian, then we have found a way of interpreting $\hat H$ as the Hamiltonian of an integrable lattice model plus a small, but potentially non-integrable correction. But we would still like to ensure that this theory yields any form of local dynamics in real space. In particular, we would like to 
\begin{itemize}
    \item[A)] assign momenta $k = k(i)$ to each of the qubit factors such that the energies $E(k) \equiv E_{i(k)}$ represent a meaningful dispersion relation;
    \item[B)] introduce Fermionic modes $\hat c_k$, such that $\lbrace \hat c_k , \hat c_{k'}^\dagger\rbrace = \delta_{k,k'}$ and $\hat Z_i = 2\hat c_{k(i)}^\dagger \hat c_{k(i)} - \mathbb{I}$;
    \item[C)] Fourier transform these momentum space modes to obtain real space modes $\hat a_x$ and express the Hamiltonian in terms of these.
\end{itemize}
For simplicity, let us assume here, that the total number of qubits is odd, \ie $n=2\ell+1$, and let us also assume that the energies $E_i$ have been ordered, \ie $E_i < E_{i+1}\ \forall i$. We then propose to achieve A) via the assignment
\begin{align}
\label{eq:assignment}
    E(k=\ell) &= E_n\nonumber \\
    E(k=\ell-1) &= E_{n-2}\nonumber \\
    \dots \nonumber \\
    E(k=1) &= E_{3}\nonumber \\
    E(k=0) &= E_{1}\nonumber \\
    E(k=-1) &= E_{2}\nonumber \\
    \dots \nonumber \\
    E(k=-\ell+1) &= E_{n-3}\nonumber \\
    E(k=-\ell) &= E_{n-1}\ ,
\end{align}
\ie we assign neighboring values $E_{i}\, \&\, E_{i+1}$ in the original ordering to opposite points $\pm k$ on the momentum grid. This ensures that the dispersion relation $E(k)$ is approximately symmetric and that it monotonically increases to either side of $k=0$.

Given such an ordering in $k$ we can proceed to define raising and lowering operators for each qubit as
\begin{equation}
    \hat b_k \equiv \frac{1}{2}(\hat X_k + i\hat Y_k)\ ,
\end{equation}
where $\hat X_k$ and $\hat Y_k$ are defined as the product operators that contain $\sigma_x$ or $\sigma_y$ in the $k$th factor and $\mathbb{I}$ in all other factors (in analogy to \eqnref{def_of_Zi_2nd}). These are still Bosonic operators and to meet condition B) from above we would like to obtain Fermionic (\ie anti-commuting) operators instead. This can be achieved via a Jordan-Wigner transformation, \ie by defining
\begin{equation}
    \hat c_k \equiv \exp\left(-i\pi \sum_{k'<k} \hat b_{k'}^\dagger \hat b_{k'}\right)\, \hat b_k\ .
\end{equation}
The operators $\hat Z_k$ (where we again use the new labels $k = k(i)\,$) are left invariant by this transformation, and they are given in terms of the new Fermionic operators as $\hat Z_k = 2\hat c_k^\dagger \hat c_k - 1\,$. Thus, in terms of the new operators the Hamiltonian becomes
\begin{equation}
\label{eq:from_H_to_field_3}
    \hat H \simeq \sum_{k=-\ell}^\ell E(k)\, \left(\hat{c_k}^\dagger \hat c_k - \frac{1}{2} \right) \ +\ \hat H_{\mathrm{int}}\ .
\end{equation}
The operators $\hat c_k\, ,\, \hat c_k^\dagger$ can now be interpreted as annihilating and creating the Fourier modes of the emergent lattice theory. From these Fourier space operators we can then define position space operators $\hat a_x$ via discrete Fourier transform,
\begin{equation}
    \hat a_x = \sum_{k=-\ell}^\ell \hat c_k\, \exp\left(i\, \frac{2\pi kx}{2\ell +1}\right)\ ,
\end{equation}
where $x$ also ranges in integer steps from $-\ell$ to $+\ell\,$. We finish this sub-section with one more corollary about the properties of the lattice theory we have now obtained.

\begin{corollary}
\label{corr:particle_number}
    If the vacuum state $\ket{0}$ is defined to be the state annihilated by all emergent Fourier space operators $\hat c_k$ or, equivalently, by all real space operators $\hat a_i$ then particle number - when defined \wrt to the emergent lattice of operators - is conserved by the Hamiltonian $\hat H$.

    Furthermore, in the basis $\ket{k} \equiv \hat c_k^\dagger \ket{0}$ of the one-particle subspace the Hamiltonian of our emergent lattice theory will always be diagonal, \ie
    \begin{equation}
        \bra{k}\hat H\ket{k'} \propto \delta_{k,k'}\ .
    \end{equation}
\end{corollary}
\begin{proof}
In Corollary~\ref{corr:Hint_from_Zs} we had seen that the interaction Hamiltonian \wrt to the emergent lattice contains only Pauli strings with either $\hat Z_k$ or $\mathbb{I}$ in each factor. Hence, each term contributing to $\hat H_{\mathrm{int}}$ contains the same number of creation and annihilation operators. Also, all of these terms will be diagonal in the one-particle sub-space \wrt to the basis $\ket{k}\,$.
\end{proof}


\begin{figure*}
\centering
  \includegraphics[width=0.75\textwidth]{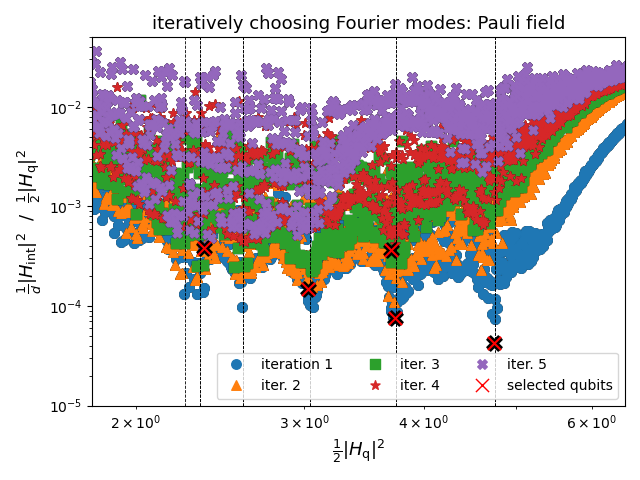}
  \includegraphics[width=0.75\textwidth]{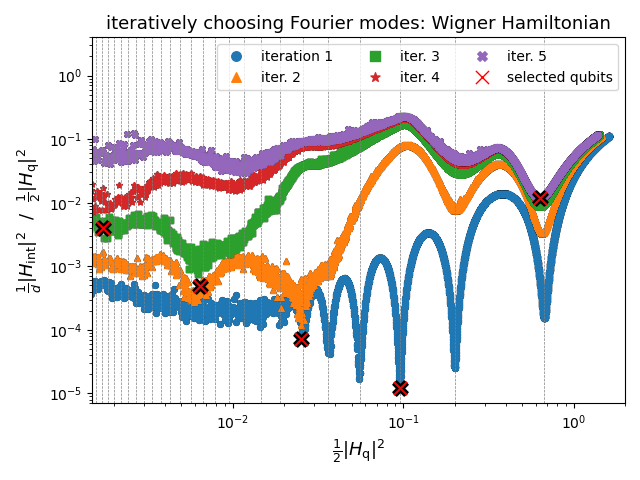}
   \caption{Visualizing our iterative procedure for factorizing a given Hilbert space into minimum-interaction qubits. At each iteration we consider the Hamiltonian $\hat H_{\mathrm{rest},i}$ that was left from the previous factorization as the new total Hamiltonian. We then sift through different norms $|\hat H_{q,i+1}|\sim E_{q,i+1}$ of the qubit self-Hamiltonian and find candidates for the rest-Hamiltonian $\hat H_{\mathrm{rest},i+1}$ such that interaction between next qubit and next rest is minimized. Among all of these candidates we then pick the qubit and rest that achieve the lowest ratio $|\hat H_{\mathrm{int},i+1}|\,/\,|\hat H_{q,i+1}|\,$. These selected factorizations are indicated by the red crosses in the figure. The upper panel displays this procedure for a Pauli field Hamiltonian, while the lower panel was starting with a random Hamiltonian drawn from the GUE. The Hilbert space dimension in both cases is $d=2^{11}\,$, \ie we do not show all iterations. Vertical dashed lines indicate qubit energies where we expect particularly low interaction strength (in the case of the Pauli field: the energies of the Fourier modes).}
   \label{fi:iterations}
\end{figure*}

\begin{figure*}
\centering
  \includegraphics[width=0.8\textwidth]{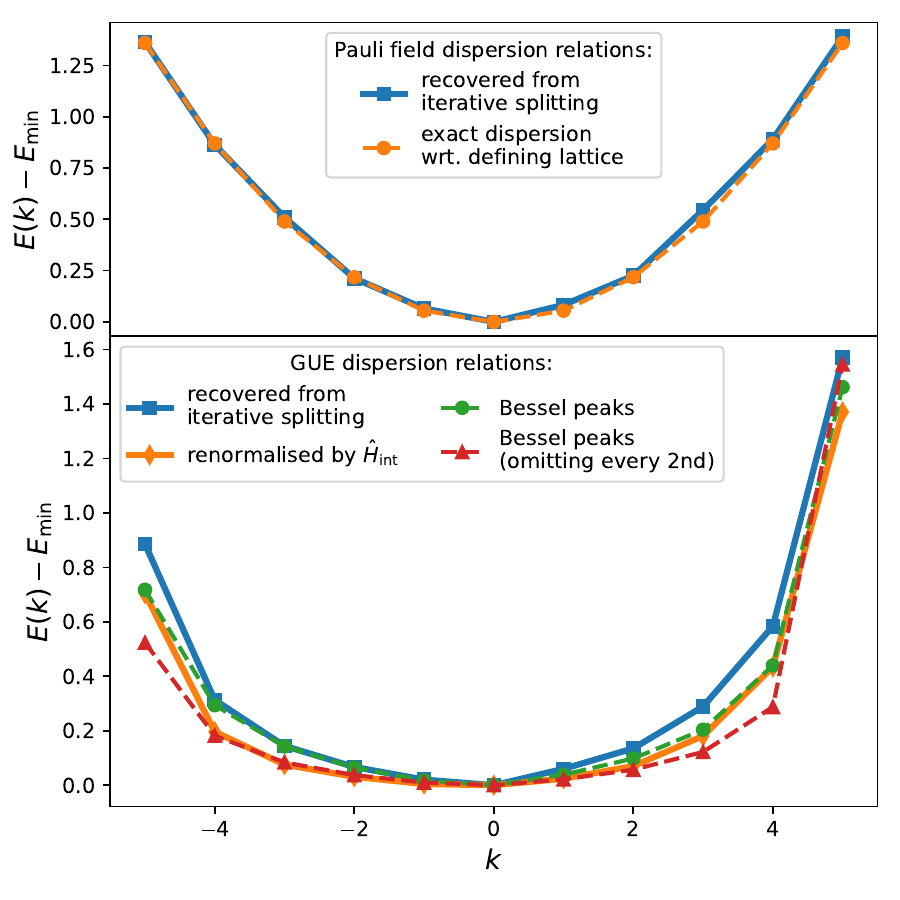}
   \caption{The dispersion relations we obtain when iteratively decomposing either the Pauli field Hamiltonian or a GUE Hamiltonian into a lattice of weakly interacting spins ("Fourier modes"). See main text for details.}
   \label{fi:dispersions_main_text}
\end{figure*}

\subsection{Application to Pauli field}
\label{sec:app_to_Pauli}

Let us first test the above algorithm with an integrable many-body Hamiltonian that has a known decomposition into a set of zero-interaction qubits. In particular, we will consider the Pauli field with Hamiltonian
\begin{equation}
\label{eq:Paul_real_space}
    \hat H = \int \dd x\ \hat a^\dagger(x) \left(m - \frac{1}{2m}\frac{\partial^2}{\partial x^2} \right) \hat a(x)\ ,
\end{equation}
where the $\hat a$ are Fermionic operators, \ie they satisfy the \emph{anti}-commutation relation $\lbrace \hat a(x), \hat a(x')\rbrace \sim \delta(x-x')\,$. The zero-interaction qubits are then exactly the Fourier modes of this field, in terms of which the Hamiltonian becomes
\begin{equation}
    \hat H \propto \int \dd k\ \hat c^\dagger(k) \left(m + \frac{k^2}{2m} \right) \hat c(k)\ .
\end{equation}
Since our algorithm works with finite-dimensional Hilbert spaces, we need to discretize this theory. A straightforward way to do so would be to replace $\partial^2/\partial x^2$ in \eqnref{Paul_real_space} by a finite-differences approximation on a finite grid. But for low dimensional grids this leads to a strongly distorted dispersion relation, and as we explain in Appendix~\ref{app:Dirac} we will use an alternative discretization scheme that is inspired by \cite{Singh2018}. Within this scheme, the Hamiltonian becomes
\begin{equation}
\label{eq:Pauli_field_H}
    \hat H = \sum_{i,j} M_{ij} \left( \hat a_i^\dagger \hat a_j - \frac{1}{2}\delta_{ij}\right)\ ,
\end{equation}
where $i$ and $j$ label the grid of real space positions, and where the eigenvalues of the coefficient matrix $M$ are exactly the Fourier mode energies $m+k^2/2m\,$ in units of the inverse grid spacing $(\Delta x)^{-1}$ .

The upper panel of Figure~\ref{fi:iterations} shows how our algorithm operates when starting from the spectrum of a Pauli field Hamiltonian with $m = 3/\Delta x$ that is defined on a grid of $n=11$ grid points. At each iteration we find candidate factorizations $\mathcal{H}_{r,i-1} \simeq \mathcal{H}_{q,i} \otimes \mathcal{H}_{r,i}$, which we display as points in the plane of $|\hat H_{\mathrm{int}}|^2 / |\hat H_q|^2$ vs.\ $|\hat H_q|^2$ , where $\hat H_q$ is the qubit self-Hamiltonian and $\hat H_{\mathrm{int}}$ is the interaction Hamiltonian in the corresponding factorization. We then pick the factorization that achieves the smallest ratio $|\hat H_{\mathrm{int}}|^2/|\hat H_q|^2\,$ --- marked by the large red crosses in Figure~\ref{fi:iterations} --- and treat the corresponding rest-Hamiltonian $\hat H_{r,i}$ as the total Hamiltonian in our next iteration.

Note however the following practical problem with our algorithm: different local minima of our initial loss function $\mathcal{L}\sim|\hat H_{\mathrm{int}}|^2$ will have different \enquote{catchment areas} and numerical approaches that search for these minima tend to strongly prefer certain minima over others. This makes it challenging to recover all relevant minima from a set of such searches. The situation is helped slightly by Theorem~\ref{th:H_structure} which tells us that the minimization can be carried out in the space of the spectra of $\hat H_q$ and $\hat H_r$, which is much smaller than the space of factorizations (see \cite{Loizeau2024} for how to implement such a minimization based on the spectra). But numerical minimization algorithms still tends to default to a restricted set of minima. We hence employ the following alternative scheme for obtaining sets of factorizations at each iteration: We sift through a set of possible spectra for the qubit self-Hamiltonians (which is equivalent to sifting through an array of values of $|\hat H_q|^2$) and find spectra for the corresponding rest-Hamiltonian $\hat H_r$ such that the norm of the interaction Hamiltonian $\hat H_{\mathrm{int}}$ is minimized (locally in the space of possible spectra of $\hat H_r$). This approach extends our coverage of minima to a wide range of possible qubit Hamiltonians. And the points in Figure~\ref{fi:iterations} were obtained in this manner.

Returning to the figure, note that there is a set of low-interaction peaks in the first iteration (blue circles). These peaks correspond to the Fourier energies of the Pauli field, \ie at the center of these peaks are factorizations whose qubit factors are exactly the Fourier mode factors of the theory. We demonstrate this with the vertical dashed lines which display the self-Hamiltonian norms $|\hat H_q|^2 \propto (m+k^2/2m)^2\,$ of these modes. The fact that we only sift through a finite grid of $|\hat H_q|^2$ prevents us from finding the exact Fourier modes. But at each iteration, our algorithm still picks and splits off a qubit that is close to one of these dashed lines. In particular, the excitation energies of the qubits that our algorithm is factoring out are in very close agreement with the exact dispersion relation $E(k) \propto m + \frac{k^2}{2m}$ of our Pauli field, as we demonstrate in the upper panel of \figref{dispersions_main_text}.

\begin{figure*}
\centering
\includegraphics[width=0.99\textwidth]{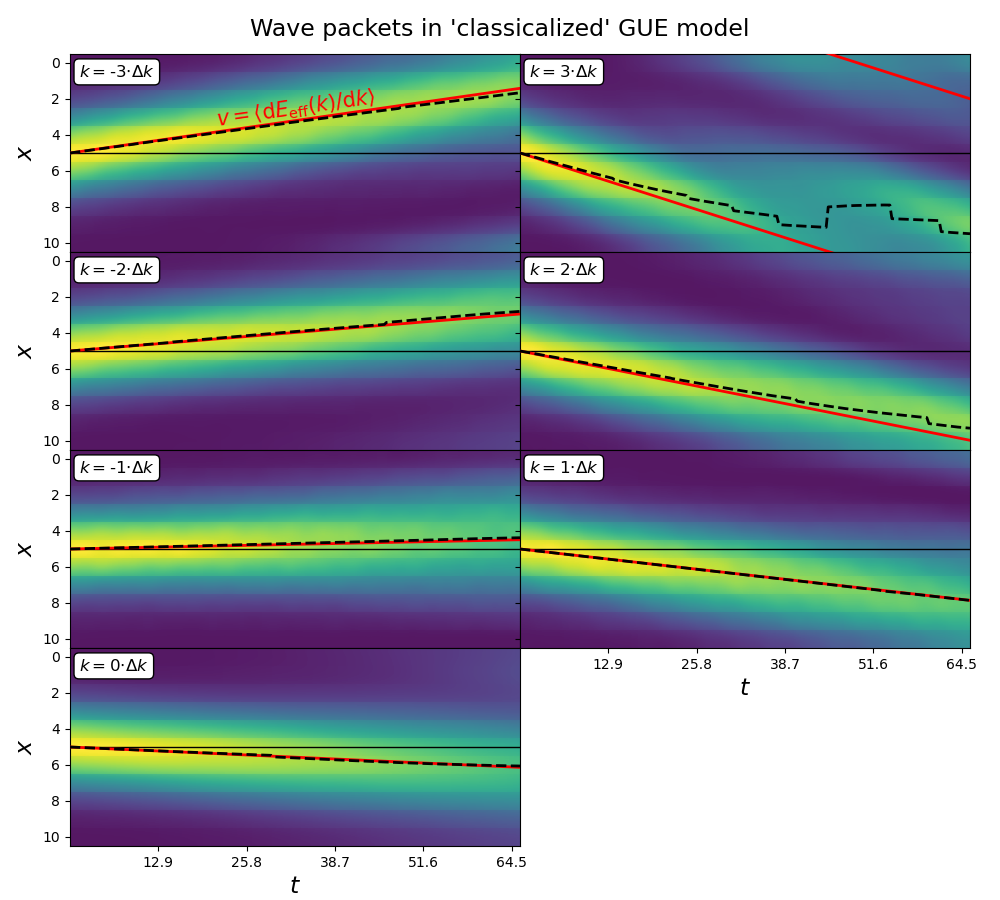}
   \caption{Starting from a Gaussian random Hamiltonian in a $d=2^{11}$ dimensional Hilbert space we identify field operators $\hat a_x^\dagger$ \wrt which stable, propagating wave packets can be defined. Different panels show how the density $|\psi(x,t)|^2$ evolves for wave packets of different average momenta $k$. Red solid lines indicate expected propagation speed from an effective dispersion relation, while black dashed lines indicate the actual center of the wave packets. Note that the packet in the upper most, right panel is dispersing quickly, because around its average momentum the dispersion relation is strongly curved. Because of the periodicity of our (small) grid, this leads to the interference pattern that is visible in that panel. We display the shape of the dispersion relation in \figref{dispersions_main_text}.}
   \label{fi:Wigner_packets}
\end{figure*}

\subsection{Application to GUE model: wave packets from random Hamiltonian}
\label{sec:app_to_GUE}

Let us now apply the same procedure to a Hamiltonian $\hat H$ that is drawn from the Gaussian Unitary Ensemble (GUE), \ie to a Wigner Hamiltonian. We present our GUE conventions in \appref{GUE}, and we also show there the following: If you pick an initial set of lattice creation and annihilation operators $\hat a_i\,,\,\hat a_i^\dagger$ and then decompose $\hat H$ into Pauli strings $\hat\tau_\ell$ \wrt that lattice,
\begin{equation}
    \hat H = \sum_\ell h^\ell\,\hat\tau_\ell\ ,
\end{equation}
(where we use the conventions stated in the beginning of this section, such that some strings contain an additional factor of $\mathrm{i}$) then drawing $\hat H$ from the GUE is equivalent to drawing the coefficients $h^\ell$ independently from a Gaussian distribution of variance
\begin{equation}
    \sigma^2(h^\ell) = \frac{1}{d}\, .
\end{equation}
This means that such an $\hat H$ will induce highly non-local dynamics \wrt the lattice-structure that was used to define the above Pauli-string decomposition. But our algorithm allows us to move away from that initial structure, and we will see that this can lead to much more well behaved dynamics.

We again consider a $d=2^{11}$ dimensional Hilbert space and iteratively identify $11$ minimum-interaction qubits. The operator norms $|\hat H_q|^2$ and $|\hat H_{\mathrm{int}}|^2$ we are sifting through in the first five of these iterations are displayed in the lower panel of Figure~\ref{fi:iterations}. Note that in the first iteration of our algorithm we find a discrete sequence of norms $|\hat H_q|$ for which the initial qubit displays a particularly low interaction with the rest of Hilbert space. We derive an understanding of these peaks in \secref{method}, and the vertical dotted lines that coincide with the peaks are a result of our calculations in that section. For now, we continue with our approach of iteratively selecting low-interaction qubits. This procedure returns a set of Fourier mode operators $\lbrace \hat c_k\rbrace$ and corresponding real space operators $\lbrace \hat a_x\rbrace$ in terms of which our Wigner Hamiltonian becomes
\begin{align}
    \hat H &= \sum_{x,y}M_{xy}\, \left(a_{x}^\dagger a_y - \frac{1}{2}\delta_{x,y}\right) + \hat H_{int} \nonumber \\
    &= \sum_{k} E(k) \left( c_{k}^\dagger c_k - \frac{1}{2}\right) + \hat H_{int} \ .
\end{align}
The blue, solid line in the lower panel of \figref{dispersions_main_text} displays the dispersion relation $E(k)$ we obtain this way. This bare dispersion relation will however be re-normalized by the interaction Hamiltonian $\hat H_{\mathrm{int}}$, which we find to be non-vanishing in the Wigner-case (as opposed to the Pauli-case we looked at in the previous subsection). This yields an effective dispersion relation
\begin{align}
    E_{\mathrm{eff}}(k) &= \bra{k} \hat H \ket{k}\nonumber \\
    &= E(k) + \bra{k} \hat H_{\mathrm{int}} \ket{k}
\end{align}
which is displayed as the solid, orange line in \figref{dispersions_main_text}. In the figure we also display different attempts at analytically estimating the dispersion relations - see \secref{method} for details. 

Taking $\ket{0}$ to be the state annihilated by all $\hat c_k$, we can build wave packet states with average momentum $k_0$, width in momentum space $\Delta$ and initial position $x_0$ like
\begin{equation}
    \ket{k_0,\Delta, x_0} \equiv \sum_k\, \mathrm{Gauss}(k|k_0,\Delta)\, e^{-ix_0 k}\,\ket{k}\ ,
\end{equation}
where $\ket{k}\equiv\hat c_k^\dagger\ket{0}$ are the single-particle plane wave states and where $\mathrm{Gauss}(k|k_0,\Delta)$ is the appropriate Gaussian one-particle wave function. With the help of the real-space \enquote{field} $\hat a_x$ we can then study the evolution of the position space wave function
\begin{equation}
    \psi(x,t) = \bra{0} \hat a_x\, \hat U(t) \ket{k_0,\Delta, x_0}\ ,
\end{equation}
where $\hat U(t) = \exp(-i\hat H t)$ describes the time evolution induced by $\hat H$. 

In Figure~\ref{fi:Wigner_packets} we display the absolute value $\rho(x,t) = |\psi(x,t)|^2$ for a number of different wave packets obtained by applying our algorithm to a Hamiltonian drawn from the GUE. In other words: we started from a Hamiltonian operator whose elements have literally been drawn at random, and we have been able to re-interpret the dynamics induced by this Hamiltonian as the dynamics of propagating particles. And the speed at which wave packets travel can be understood via the above mentioned effective dispersion relation in the usual way, i.e.\ $v = \langle \dd E_{\mathrm{eff}}/\dd k\rangle\,$.

In the light of Corollary~\ref{corr:particle_number} the results of \figref{Wigner_packets} are not surprising: \emph{any} Hamiltonian will - when re-interpreted as a lattice theory via our algorithm - preserve particle number and also be diagonal in the momentum basis of the one-particle Hilbert space. Hence - for any input Hamiltonian - the one-particle Hilbert space of our emergent lattice theory is guarantied to host wave functions that propagate according to the effective dispersion relation $E_{\mathrm{eff}}(k)\,$. Whether or not that new interpretation of $\hat H$ resembles a classical system will then only depend on the shape of $E_{\mathrm{eff}}(k)\,$. In the case of the GUE-example, the shape of $E_{\mathrm{eff}}$ enables (for the most part) localized wave packets to propagate and stay local. However, some artifacts are visible from \figref{Wigner_packets}: Firstly, note that $k=0$ is not the slowest wave packet. Also, note that an interference pattern starts to form in the upper most, right panel. The packet in that panel is centered around momenta where $E_{\mathrm{eff}}$ is strongly curved. Thus the packet disperses quickly and starts to meet itself from the other side of the periodic grid. We would expect both of these artifacts to become less severe in the limit of larger Hilbert space dimension.

Finally, note that we do not investigate here the interaction structure that $\hat H_{\mathrm{int}}$ induces \emph{between} particles in the higher $N$-particle sub-spaces.

\subsection{MDS embedding of the emergent lattices}

After we apply our sifting procedure, any Hamiltonian will take the form
\begin{equation}
    \hat H = \sum_{x,y}M_{xy}\, \left(a_{x}^\dagger a_y - \frac{1}{2}\delta_{x,y}\right) + \hat H_{int}
\end{equation}
in terms of the emergent real-space creation and annihilation operators $\hat a_{x}^\dagger\,,\, \hat a_x\,$. And the elements of the coefficient matrix $M_{xy}$ will be given in terms of the emergent dispersion relation $E(k)$ as
\begin{equation}
\label{eq:M_xy_comp}
    M_{xy} = \frac{1}{(2\ell+1)^2} \sum_{k=-\ell}^{\ell} E(k)\,  \exp\left(i\,\frac{2\pi k(x-y)}{2\ell+1}\right)\ .
\end{equation}
This only depends on the difference of $x$ and $y$, \ie the quadratic part of the Hamiltonian for our emergent lattice theory is translation invariant - regardless of the specifics of the total Hamiltonian we started with. Furthermore, $M_{yx} = M_{xy}^*$ such that we can consider $|M_{yx}|^2$ as a matrix of \enquote{hopping} terms in real-space with periodic boundary conditions. Following \cite{Cao2017} we use this matrix to define a metric between the lattice sites. We first calculate $D_{xy}=-\log(M_{xy}+\epsilon)\,$, with $\epsilon=10^{-12}$ ensuring that $D$ is finite. This defines a matrix of distances on the graph of lattice sites, from which we can obtain a metric by looking for the shortest path between any two sites,
\begin{equation}
\label{eq:MDS_metric}
    m_{xy}=\min_{\gamma}\left(\sum_{n=0}^{k_{\gamma}-1}D_{\gamma(n),\gamma(n+1)}\right)
\end{equation}
for all connected paths $\gamma$ and $k_{\gamma}$ the length of a path. Based on the above metric and using multidimensional scaling (MDS) as e.g.\ described in \cite{Cao2017}, we now look for the optimal embedding of the graph of lattice sites into lower dimensional Cartesian spaces. In particular, we want to test whether that optimal embedding is consistent with the supposed 1D nature of our emergent lattice.

\begin{figure*}
    \centering
    \includegraphics[width=0.9\linewidth]{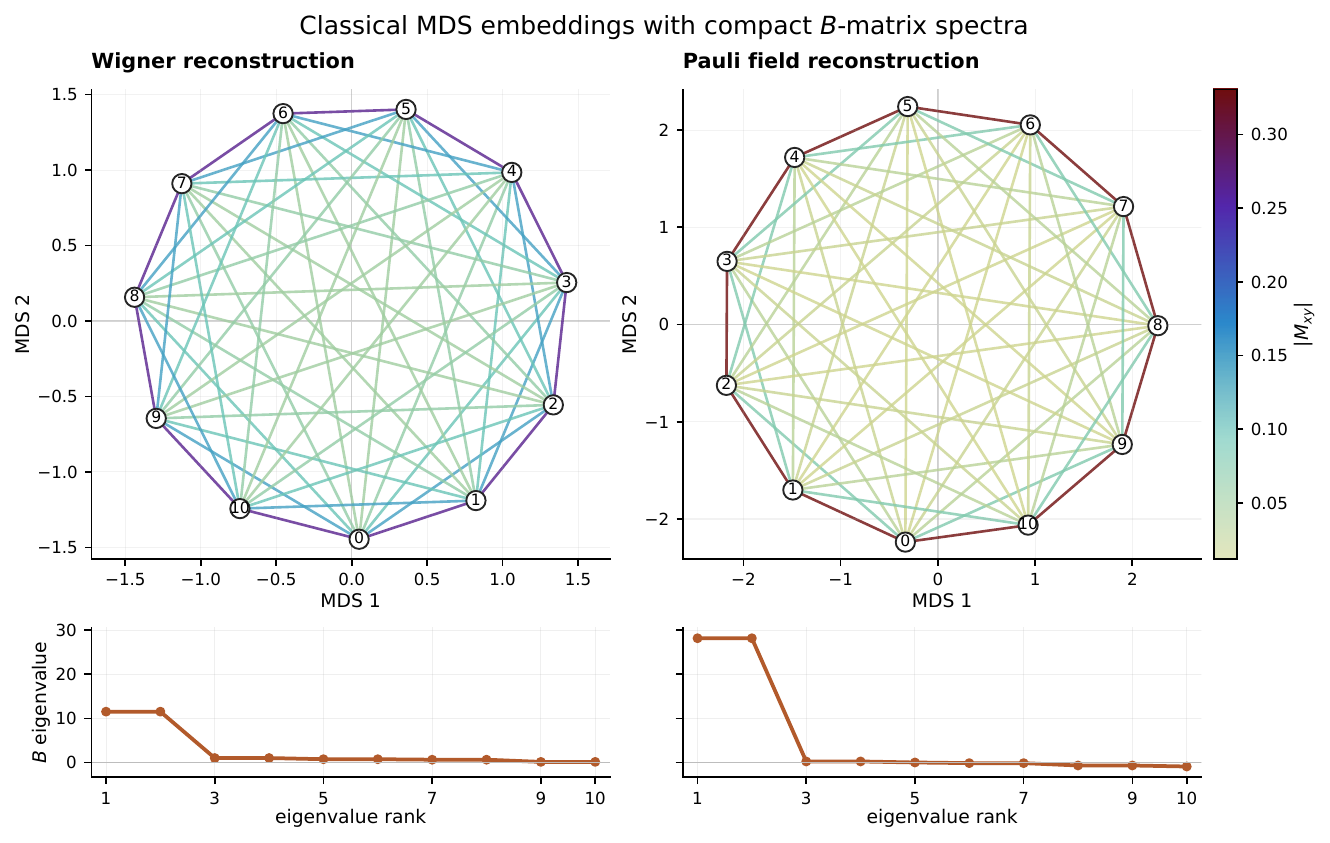}
    \caption{Classical MDS embeddings with compact $B$-matrix spectra: The top left panel shows the embedding of the spatial Fermionic degrees-of-freedom emerging from a Wigner Hamiltonian in a two-dimensional cartesian space via MDS. The top right panel shows the corresponding embedding for a Pauli field Hamiltonian. We overlay graph edges between embedded Fermionic degrees-of-freedom, with edge color indicating hopping term strength. Embedding coordinates MDS 1 and MDS 2 emerge as the dominant Cartesian dimensions of the embedding, as is implied by the bottom panels. These panels show the ordered spectrum of the $B$ matrix for Wigner (left) and Pauli field reconstruction (right). In classical MDS, $B$ is simply the double-centered squared distance matrix (where distance is calculated via the metric from \eqnref{MDS_metric}). Its eigenvalues quantify the spatial variance of the data along its Cartesian eigendimensions. $\lambda_1\approx \lambda_2> \lambda_{i>2}$ indicates MDS 1 and 2 are equally important and dominant compared to the others. The difference in scale between the Wigner- and Pauli-case comes from the fact that the emerged distances are different (larger for the Pauli field), due to differences in coupling strength.} 
    \label{fi:MDS_Embedding}
\end{figure*}

In both the Pauli field and GUE examples from above MDS should indeed be applicable to the hopping matrix $|M_{yx}|^2$ since in both cases we find $|H_{int}|$ to be of the order of a percent of the norm of the total Hamiltonian $\hat H$, implying that quadratic terms dominate over higher order terms of the form $|b_{ijk\cdots}|=|\mathrm{Tr}[Ha_ia_j^\dagger a_k\cdots]|$. The upper panel of \figref{MDS_Embedding} show's a 2D embedding of the distance matrix, extracted from the Wigner Hamiltonian (left) and from the Pauli field Hamiltonian. Unsurprisingly, the cyclicity in the hopping components (\cf \eqnref{M_xy_comp}) means that the optimal embedding of our construction in 2D is as a circle. The lower panel of \figref{MDS_Embedding} also hints at a circle being the natural emergent geometry with both first eigenvalues of the double-centered, squared distance matrix being approximately equal. The two-dimensional flatness of the data is indicated by all other eigenvalues being close to zero.


\section{Sequence of candidate degrees-of-freedom from the GUE spectrum}
\label{sec:GUE_species}
\label{sec:method}

\begin{figure*}
\centering
\includegraphics[width=0.9\textwidth]{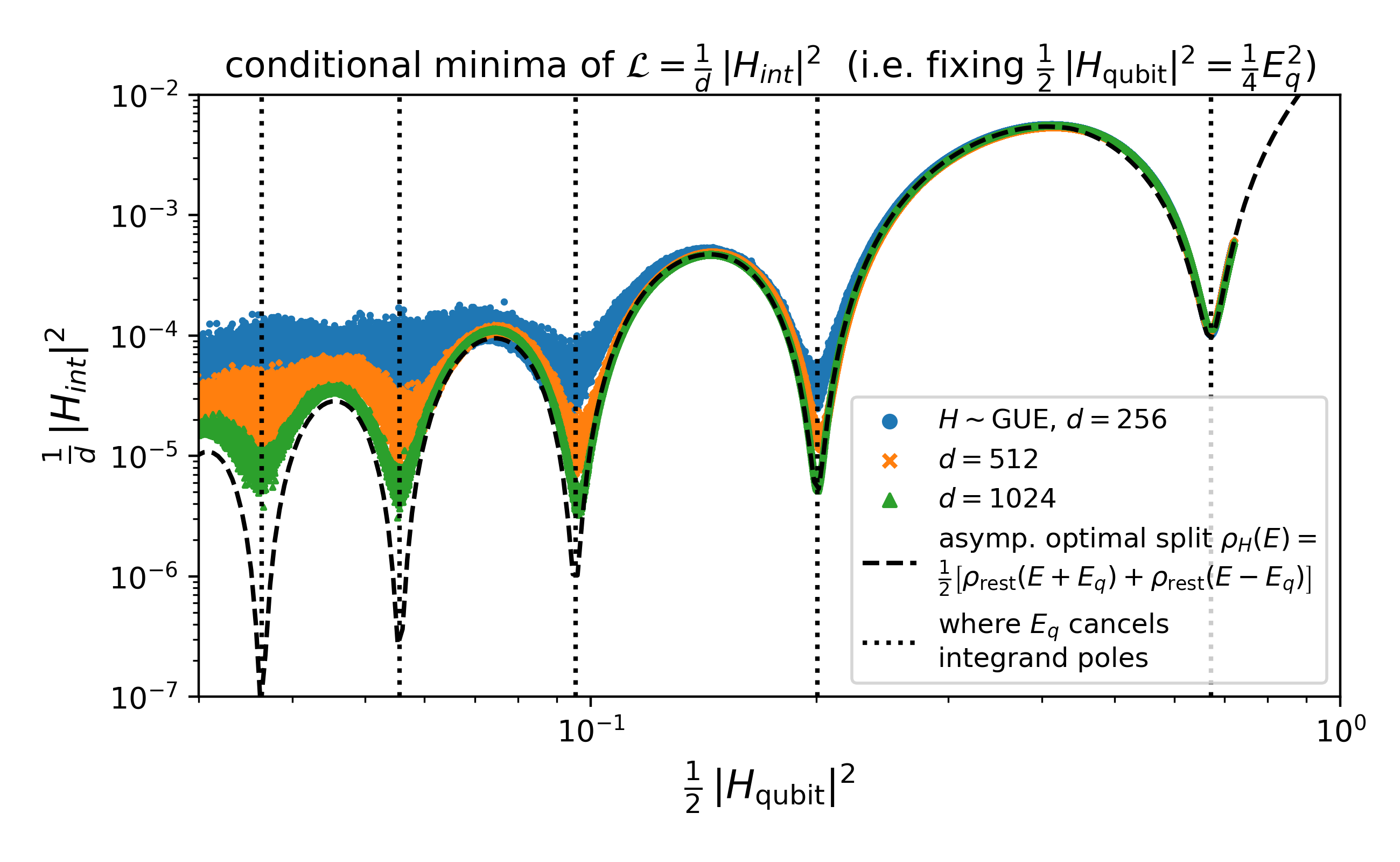}
   \caption{We show the squared norm of $\hat H_{\mathrm{int}}$ when fixing the energy gap $E_q = E_{\max} - E_{\min}$ of the qubit Hamiltonian (and thus fixing the norm of $\hat H_q$). These conditional minima closely follow the asymptotic lower bound we derived in \secref{method}. The vertical, dotted lines indicate the location of minimum-interaction-peaks that is predicted by \eqnref{minimum_int_peaks}.}
   \label{fi:minimum_distribution}
\end{figure*}

In the lower panel of \figref{iterations} we saw that there is a discrete set of qubit energies $E_q \sim |\hat H_q|$ for which the initial qubit selected by our algorithm would have particularly low interaction with the rest of Hilbert space. In this section we will derive an analytic understanding of that behavior based on the asymptotic eigenvalue density of GUE matrices. This will also help us to qualitatively model the dispersion relation of the lattice theory that emerged from the GUE Hamiltonian in \secref{app_to_GUE}. The main findings of this section can be summarized as follows:\\

\noindent \textbf{Result A.} \emph{There is an asymptotic ($d \rightarrow \infty$) lower bound for $\frac{1}{d}\mathrm{Tr}(\hat H_{\mathrm{int}}^2)$ as a function of the qubit excitation energy $E_q$, and this bound can be computed from the asymptotic GUE eigenspectrum.}
\\

\noindent \textbf{Result B.} \emph{There is a discrete set of excitation energies $E_q$ where the above bound dips particularly low, \ie where factorizations with particularly low interaction strength are possible. These energies $E_q$ can also be computed from the asymptotic GUE eigenspectrum.}
\\

\noindent To derive these results, let us suppose that there is indeed a factorization $\mathcal{H} \simeq \mathcal{H}_q \otimes \mathcal{H}_r$ for which the interaction part of the Hamiltonian is negligible compared to the auto-Hamiltonians $\hat H_q$ and $\hat H_r\,$. For such a factorization we should have 
\begin{align}
\label{eq:H_approx}
    \hat H\, \approx\, \frac{E_q}{2}\, \sigma_z \otimes \mathbb{I}_r\, +\, \mathbb{I}_q \otimes \hat H_{r}\ ,
\end{align}
where we again chose a basis in $\mathcal{H}_q$ such that $\hat H_q \propto \sigma_z\,$.

In order for a given factorization to satisfy the approximate \eqnref{H_approx}, the spectral density $\rho_{\mathrm{total}}$ of the total Hamiltonian should be a convolution of the spectral densities $\rho_{q}$ and $\rho_{r}$ of the self-Hamiltonians of the qubit and the rest, \ie
\begin{align}
\label{eq:relation_of_densities}
    \rho_{\mathrm{total}}(\lambda) &\approx \left(\rho_{q} * \rho_{r}\right)(\lambda) \nonumber \\
    &= \frac{1}{2}\left[\,\rho_r(\lambda - E_q/2) + \rho_r(\lambda + E_q/2)\,\right]\ .
\end{align}
Now let $\varphi_{\mathrm{total}}(k)$ and $\varphi_r(k)$ be the characteristic functions of $\rho_{\mathrm{total}}(\lambda)$ and $\rho_{r}(\lambda)$ respectively (i.e.\ the Fourier transforms of $\rho_{\mathrm{total}}$ and $\rho_{r}$). Then Fourier transforming \eqnref{relation_of_densities} leads to
\begin{align}
    \varphi_{\mathrm{total}}(k) &\approx \frac{1}{2}\varphi_{r}(k)\left[e^{-i\frac{k E_q}{2}} + e^{i\frac{k E_q}{2}}\right]\nonumber \\
    &= \varphi_{r}(k) \cos(k E_q/2) \nonumber \\
    \label{eq:invert_convolution_equation}
    \Rightarrow \varphi_{r}(k) &\approx \frac{\varphi_{\mathrm{total}}(k)}{\cos(k E_q/2)}\ .
\end{align}
This seems to indicate that for any value of $E_q$ we can find a rest-Hamiltonian $\hat H_r$ such that \eqnref{H_approx} is satisfied exactly by simply tuning the spectral distribution of $\hat H_r$ to be given by the inverse Fourier transform of \eqnref{invert_convolution_equation}, \ie
\begin{equation}
    \rho_r(\lambda) = \int \dd k\, e^{ik\lambda}\, \frac{\varphi_{\mathrm{total}}(k)}{\cos(k E_q/2)}\ .
\end{equation}
However, this may not be a proper probability density function (PDF). E.g.\ it may not even satisfy $\rho_r(\lambda) \geq 0$ for all $\lambda$. And how close $\rho_r$ comes to being an actual PDF may depend on $E_q\,$.

\begin{figure*}
\centering
\includegraphics[width=\textwidth]{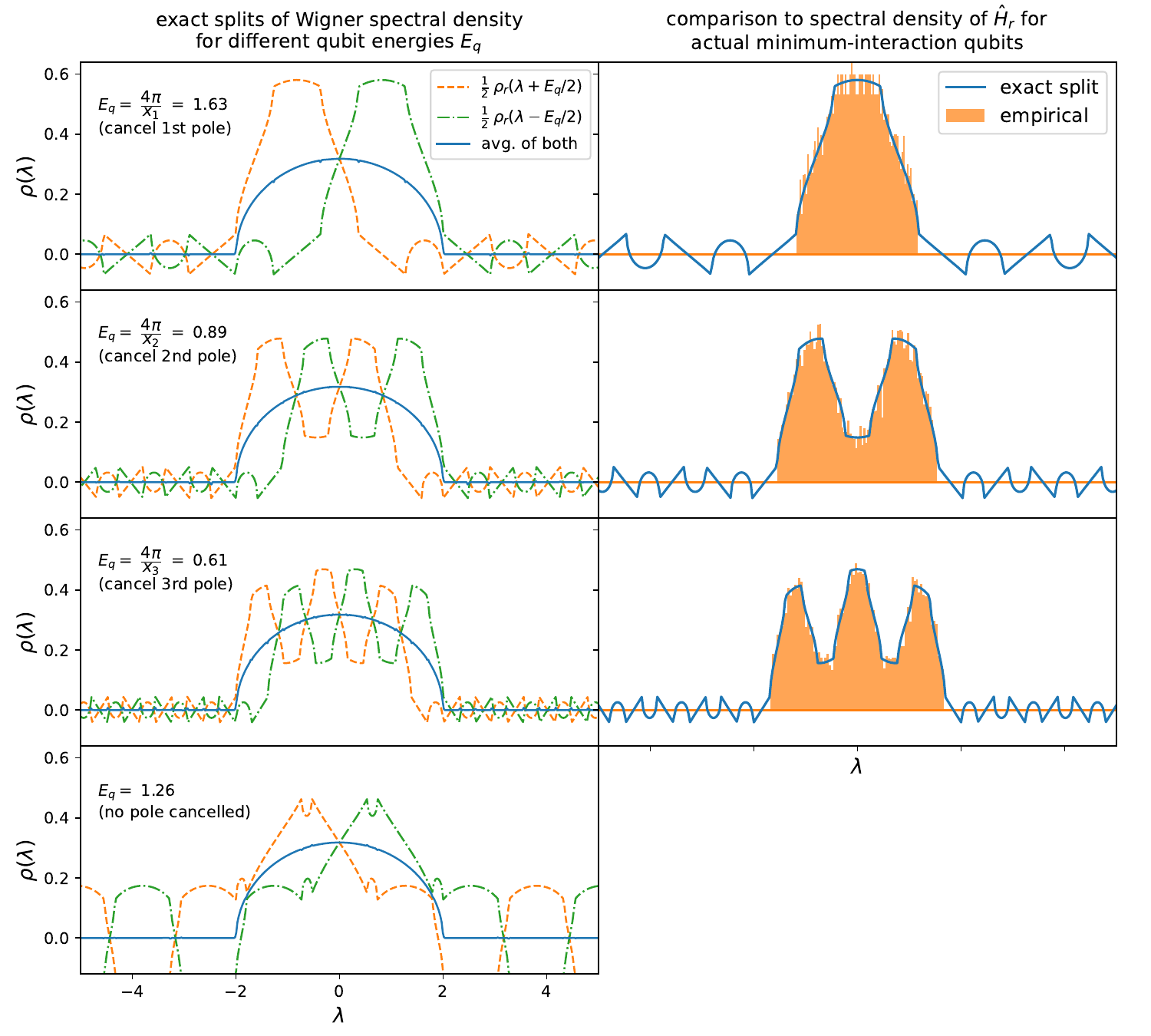}
   \caption{Left panels are showing exact solutions to the decomposition of \eqnref{relation_of_densities} for different energy gaps $E_q$ of the qubit Hamiltonian. Note that these solutions do not represent actual densities and can thus not be implemented by the spectral density $\rho_r$ of any rest-factor Hamiltonian $H_r\,$. But for some values of $E_q$ - cf.\ the upper three panels and \eqnref{minimum_int_peaks} - the exact solutions are well approximated by actual densities over the energy range relevant for the decomposition. This is demonstrated in the right panels where we compare these solutions to the actual spectral densities $\rho_r$ of the rest-factor Hamiltonians found for minimum-interaction qubits whose energy gap $E_q$ approximately satisfies \eqnref{minimum_int_peaks}.}
   \label{fi:finger_functions}
\end{figure*}

Since we are considering the total Hamiltonian to be drawn from the GUE, its eigenvalue density will tend (in distribution) towards the Wigner semicircle distribution
\begin{equation}
    \rho_{\mathrm{total}}(\lambda) \approx \frac{1}{2\pi} \sqrt{4 - \lambda^2}
\end{equation}
in the limit $d\rightarrow \infty\,$. The Fourier transform of this is
\begin{equation}
    \varphi_{\mathrm{total}}(k) \equiv \int \frac{\dd\lambda}{2\pi}\, \rho_{\mathrm{total}}(\lambda)\, e^{-ik\lambda} = \frac{J_1(2k)}{2\pi k}\ ,
\end{equation}
where $J_1$ is the first order Bessel function. Splitting $\hat H$ into a qubit-Hamiltonian and a rest-Hamiltonian should then lead to a spectral density of that rest-Hamiltonian that is given by
\begin{equation}
\label{eq:exact_split}
    \rho_r(\lambda) \approx \int \dd k\, e^{ik\lambda}\, \frac{J_1(2k)}{2\pi k \cos(k E_q/2)}\ .
\end{equation}
The right hand-side of this equation is not a purely positive function and is thus not a proper density - except for $E_q = 0$, in which case the emergent qubit would have trivial self-dynamics. 

To find approximate solutions to \eqnref{relation_of_densities} that are proper densities, let us re-interpret our loss function $\mathrm{Tr}(\hat H_{\mathrm{int}}^2)$ as
\begin{align}
\label{eq:loss_at_minimum}
    \mathcal{L}(\hat H_{\mathrm{int}}) \approx\sum_i [\lambda_i(\hat H) - \lambda_i(\hat H_{q} \otimes \mathbb{I}_r\, +\, \mathbb{I}_q \otimes \hat H_{r})]^2\ ,
\end{align}
where $\lambda_i(\hat X)$ is the $i$th ordered eigenvalue of the operator $\hat X\,$. The right hand-side of \eqnref{loss_at_minimum} is not generally equivalent to our loss function, but thanks to Theorem \ref{th:H_structure} it indeed equals $\mathrm{Tr}(\hat H_{\mathrm{int}}^2)$ at the local minima of that loss function. Now let $E_q$ be the qubit energy of a local minimum and let $\rho_r$ be the spectral density of the corresponding rest Hamiltonian $\hat H_r\,$. Then $\lambda_i(\hat H_{q} \otimes \mathbb{I}_r\, +\, \mathbb{I}_q \otimes \hat H_{r})$ is approximately given by
\begin{align}
\label{eq:lambda_quantile_qr}
    &\lambda_i(\hat H_{q} \otimes \mathbb{I}_r\, +\, \mathbb{I}_q \otimes \hat H_{r}) \approx\nonumber \\
    &\frac{1}{1/d} \int_{l_{i-1}}^l \dd l\ l\ \frac{1}{2}\left(\rho_r(l - E_q/2) + \rho_r(l + E_q/2)\right)\ ,
\end{align}
where the interval $[l_{i-1}, l_i]$ is the $i$th of a number of $d$ equal-probability quantiles of the distribution $\frac{1}{2}\left(\rho_r(l - E_q/2) + \rho_r(l + E_q/2)\right)\,$. Similarly, we have
\begin{equation}
\label{eq:lambda_quantile_H}
    \lambda_i(\hat H) \approx \frac{1}{1/d} \int_{r_{i-1}}^r \dd r\ r\ \frac{1}{2\pi} \sqrt{4 - r^2}\ ,
\end{equation}
where the interval $[r_{i-1}, r_i]$ contains the $i$th equal-probability quantile of the semi-circle distribution. We can now plug Equations~\ref{eq:lambda_quantile_qr} and \ref{eq:lambda_quantile_H} into \eqnref{loss_at_minimum} and minimize this loss function subject to the conditions $\rho_r(\lambda) \geq 0$ and $\int \dd\lambda\, \rho_r(\lambda) = 1\,$. In practice, we do this by expanding \eqnref{loss_at_minimum} to 2nd order in the deviation between $\rho_r$ and the exact, unconstrained solution from \eqnref{exact_split} and then solving the resulting constrained, quadratic problem. We provide a notebook within the \verb|GPUniverse| code package that solves this problem numerically (\url{https://github.com/OliverFHD/GPUniverse}).

The solution to the above quadratic problem for a given fixed $E_q$ gives us an estimate of the lowest loss that can be asymptotically achieved in any Hilbert space factorization with qubit energy $E_q$. The black dashed line in \figref{minimum_distribution} displays this optimum loss as a function of $\frac{1}{2}\mathrm{Tr} \hat H_q^2 = \frac{1}{4}E_q^2$. We also display the actual values of $\mathrm{Tr}(\hat H_{\mathrm{int}}^2)$ we obtain when sifting through $E_q$ and finding the corresponding rest-spectrum that minimizes interaction (as we had already done for the lower panel of \figref{iterations}). Different symbols show these minimum interaction strengths obtained for Hamiltonians drawn from the GUE in different dimensions.

As can be seen in \figref{minimum_distribution}, the lower bound we derived displays a spiked pattern, \ie there seems to be a discrete set of qubit energies $E_q$ for which particularly low loss (\ie particularly low interaction between the qubit and the rest of Hilbert space) can be achieved. This happens because for certain values of $E_q$ the Bessel function in \eqnref{exact_split} cancels one of the poles of the integrand, thus suppressing the otherwise strongly oscillating behavior of the solutions to the integral. This makes the exact solutions to the unconstrained problem of \eqnref{relation_of_densities} closer to being proper distributions, thus leading to a lower loss also for the solution to the constrained problem. We demonstrate this in the left-hand side of \figref{finger_functions} where the upper 3 panels show exact solutions to \eqnref{relation_of_densities} with $E_q$ chosen such that a pole cancellation happens, while the lowest panel is using a random value for $E_q\,$. From \eqnref{exact_split} it is clear that a pole cancellation will happen when
\begin{equation}
\label{eq:minimum_int_peaks}
    \frac{E_q}{2} \approx \frac{2\pi}{x_i}\ ,
\end{equation}
where $x_i$ is the $i$th zero of the Bessel function $J_1\,$. We display these values for $E_q$ as vertical dotted lines in the lower panel of \figref{minimum_distribution}, and they indeed coincide with the downward spikes in the interaction-lower-bound.

For $d=512$ we also compare the actual eigenvalue distribution of $\hat H_r$ at local minima in three different spikes to the corresponding exact solutions of \eqnref{relation_of_densities} in the right-hand side of \figref{finger_functions}. The approximate agreement between the latter and the former is further confirmation of our above interpretation of the spikes.

We conclude this section by returning to the lower panel of Figure~\ref{fi:iterations}. For the first iteration of our algorithm we had found there the same low-interaction-peak structure that we now encountered in \figref{minimum_distribution}. And as explained above, the location of these peaks corresponds exactly to the qubit energies $E_q$ where one of the poles of the integrand in \eqnref{exact_split} is canceled by the Bessel function that appears there. Some of these \enquote{Bessel peaks} persist in later iterations, and a number of the iteratively identified qubits are indeed chosen from a sub-set of the initial peaks (\cf some of the red crosses in the figure). This seems to indicate that qubits from these peaks are compatible in the sense that they live in each other's \enquote{rest-part} $\mathcal{H_{\mathrm{rest}}}$ of Hilbert space. But notably not all peaks persist: \eg the peak at $\frac{1}{2}|\hat H_q|^2\sim 10^{-1}$ vanishes already after the first iteration despite that iteration not selecting a qubit from said peak. We argue in \secref{towards_field_mereology} that this means that some of the peaks represent \emph{overlapping degrees-of-freedom} that don't live on mutually compatible Hilbert space factors. Such overlapping dofs have \eg been considered by \cite{Friedrich2024} (see also \cite{Chao2017, Cao:2023gkw}) to model quantum fields that satisfy a holography-inspired entropy bound. And Figure~\ref{fi:iterations} may indeed indicate that we should consider a set of emergent overlapping qubits rather than a set of perfectly separated ones in order to optimally approximate a close-to-classical lattice theory. 

We postpone a discussion of this to \secref{towards_field_mereology} and for now use our understanding of the low-interaction peaks for attempts at modeling the dispersion relation $E(k)$ that we found for Figure~\ref{fi:dispersions_main_text}. A first attempt may be to simply pick the first $n=11$ qubit energies $E_q$ for which a zero of the Bessel function cancels an integrand pole. The energies $\lbrace E_q \rbrace$ may then be sorted according to the assignment scheme we presented around \eqnref{assignment} to obtain the corresponding dispersion relation. This is displayed as the green dashed line in the lower panel of \figref{dispersions_main_text}. But as we have seen, our iterative scheme does not pick qubits from every of the initial Bessel peaks. So with the red dashed line we also display the dispersion that is obtained by selecting qubits only from every 2nd peak. Neither of these approximations perfectly matches the $E(k)$ that was actually returned by our algorithm. This is in part because omitting every 2nd Bessel peak is too simplistic and does not accurately describe which peaks are and are not selected by our algorithm. And it is in part because at any finite dimension the set of minimum-interaction qubits only resolves a finite subset of the peaks - see \eg \figref{minimum_distribution}. Nevertheless, the approximation qualitatively captures the behavior of $E(k)$ and should become more accurate for higher dimensions. Of course, $E(k)$ is only the bare dispersion relation of our emergent lattice theory. Characterizing the shape of the re-normalized dispersion $E_{\mathrm{eff}}(k)$ is a major open task required in order to \textit{a priori} understand the dynamics of that theory.


\section{Discussion: towards quantum mereology for fields}
\label{sec:discussion}
\label{sec:towards_field_mereology}

With this work we wanted to investigate how quantum mereology can be formulated such that it returns field-like degrees-of-freedom. The algorithm we propose starts with a Hamiltonian $\hat H$ on a Hilbert space $\mathcal{H}$ of dimension $d=2^n$ and identifies $n$ Fermionic creation and annihilation operators $\hat c_k^\dagger\,,\, \hat c_k$ such that the degrees-of-freedom represented by these operators have only small interaction. We then treat these operators as the Fourier modes of an emergent lattice theory and, after Fourier transforming, arrive at a lattice of operators $\hat a_x^\dagger\,,\, \hat a_x$ in terms of which $\hat H$ takes the form
\begin{equation}
\label{eq:Hxy_in_discussion}
    \hat H = \sum_{x,y}M_{xy}\, \left(\hat a_{x}^\dagger \hat a_y - \frac{1}{2}\delta_{x,y}\right) + \hat H_{int}\ .
\end{equation}
Theorem~\ref{th:H_structure} and Corollaries~\ref{corr:H_structure}-\ref{corr:particle_number} describe general properties of this emergent lattice theory that hold regardless of the input Hamiltonian $\hat H\,$. In particular, we showed that the one-particle Hilbert space of the theory will always be a theory of wave functions that propagate according to an effective dispersion relation $E_{\mathrm{eff}}(k)$ whose shape depends on the spectrum of $\hat H\,$.

Our algorithm was able to correctly recover the defining lattice structure of a Pauli field Hamiltonian. Encouraged by this consistency test, we went on to apply the algorithm to a Hamiltonian that exhibits extreme non-locality \wrt some initial lattice structure (random Gaussian coefficients in front of \emph{all} possible Pauli strings). We were able to switch to a new lattice structure that turns this Hamiltonian into a theory of particles that can propagate and stay localized for some time. 

The above results can be considered a small success in the direction of quantum mereology for fields. But a large number of questions still need to be addressed in order to formulate a map from Hamiltonian spectra to closest-to-classical (lattice) field theories. In the following, we close by summarizing some particularly important open points.
\begin{itemize}
    \item \textbf{Loss function}
    
    Throughout this paper we have defined our candidate degrees-of-freedom as those qubit factors that minimize the norm of the interaction Hamiltonian. This is only one of many choices for defining a mereology loss function, see \eg the overview of different mereology approaches we gave in \secref{intro}. Note in particular, that we applied our loss function at each individual step of the iterative procedure we proposed in Section~\ref{sec:iterative_algorithm}. One may \eg attempt to replace that with a global loss function that judges the full decomposition of Hilbert space into the full set of degrees-of-freedom. Such a loss function may \eg measure the degree of locality of the interaction Hamiltonian \wrt the emergent lattice.

    \item \textbf{Harmonic vs. real space modes}

    We have treated the minimum interaction qubits that are returned by our procedure as the Fourier modes of an emergent lattice theory. This is inspired by the fact that Fourier modes are indeed the minimum-interaction (and actually: zero-interaction) degrees-of-freedom of free field theories / of integrable lattice models. But for theories with sufficiently strong interaction, this may not be the case anymore. It was \eg shown by \cite{Piazza2010}, that for certain types of interacting lattice theories the real-space lattice modes constitute the minimum-interaction degrees-of-freedom.
    
    \item \textbf{Beyond spin-$\frac{1}{2}$ lattice}

    The field theories of our experience contain Bosonic and Fermionic fields of different spins, and interactions between those fields. From a mereology point-of-view this raises the question whether a given quantum theory - as \eg characterized by its Hamiltonian spectrum or any other data - can be factorized into locally interacting lattices of different spins, and which of these spin lattices comes closest to resembling a local field theory. For Bosonic lattices, this will only be the case if the Hilbert space considered is infinite dimensional (which may require significant extension of the framework we discussed here), or if the Bosonic degrees-of-freedom are approximated by finite dimensional constructions as has \eg been suggested by \cite{Singh2018, Cao2019_essay, Friedrich2022}.

    \item \textbf{Optimal lattice dimension}

    We have only considered turning a given quantum theory into a 1-dimensional lattice theory. In particular, around \eqnref{assignment} we mapped the self-energies of our emergent qubits onto a 1-dimensional dispersion relation. This mapping can easily be extended to dispersion relations in higher dimension, and for a given Hamiltonian spectrum this mapping may work better (in the sense of returning a closer approximation to a local field theory) in some dimensions than in others.
    
    \item \textbf{Overlapping degrees-of-freedom}

    In \secref{method} we showed that GUE Hamiltonians asymptotically allow for an entire sequence of low-interaction qubits. But when we applied our factorization procedure iteratively, some elements of this sequence where omitted. This indicates that not all of the peaks in Figure~\ref{fi:minimum_distribution} represent mutually compatible degrees-of-freedom. In other words: any two peaks in that figure represent two low-interaction factorization
    \begin{align}
        \mathcal{H} &\simeq \mathcal{H}_{q_1}\otimes \mathcal{H}_{r_1}\nonumber \\
        &\simeq \mathcal{H}_{q_2}\otimes \mathcal{H}_{r_2}\nonumber
    \end{align}
    but in general
    \begin{equation}
        \mathcal{H} \not\simeq \mathcal{H}_{q_1}\otimes\mathcal{H}_{q_2}\otimes \mathcal{H}_{r_{12}}\nonumber\ .
    \end{equation}
    This is reminiscent of the concept of \enquote{overlapping degrees-of-freedom}\cite{Chao2017, Cao:2023gkw, Friedrich2024}. In \cite{Friedrich2024} some of us had used such dofs to build a quantum field that satisfies a holographic maximum-entropy bound.

    \hspace{\parindent} The algorithm we proposed here for extracting lattice theories from a given Hamiltonian spectrum is not allowing for such overlapping dofs and instead insists on identifying an exact factorization. This means that some of the dofs we extract will have stronger interaction than the low-interaction peaks that were omitted by our algorithm. Allowing a mereology algorithm to extract sets of overlapping dofs may thus yield closer approximations to local lattice theories (as was already speculated in \cite{Friedrich2024}).

    \hspace{\parindent} Note also the following: in a regime where minimum-interaction dofs are the real-space modes of the emergent lattice theory (\cf our discussion of the results of \cite{Piazza2010} above) then the mode overlaps would represent kinematical non-locality like it is \eg present in gravitationally dressed operator fields \cite{donnelly+giddings2016_1}.
    
    \item \textbf{Dynamical structure vs.\ information structure}
    
    We proposed an algorithm for re-interpreting a given Hamiltonian as approximately local dynamics \wrt a set of field-like degrees-of-freedom. And to achieve that we indeed only required knowledge of the Hamiltonian, \ie we did not need to specify a state. This will be insufficient in the context of quantum gravity, where it seems to be the mutual quantum information that a given state induces between Hilbert space factors that determines which of these factors host spacetime (or \enquote{IR}) degrees-of-freedom and which host matter (or \enquote{UV}) degrees-of-freedom \cite{Cao2017, CaoCarroll2018, Jacobson2016} . In Appendix~\ref{app:Wishart} we explore how our algorithm could be modified to apply to states and their quantum information structure (though with very limited scope).
    
    \item \textbf{Remaining ambiguities in defining $\hat c_k^\dagger\,,\,\hat c_k$}

    We had seen in Corollary~\ref{corr:Hint_from_Zs}, that our emergent field operators $\hat c_k^\dagger\,,\,\hat c_k$ appear in the Hamiltonian only via the combination
    \begin{equation}
        \hat Z_k = 2\, \hat c_k^\dagger \hat c_k -1\ .
    \end{equation}
    A consequence of this is that the Hamiltonian spectrum together with our algorithm really only fixes the operators $\hat Z_k$ of the emergent lattice dofs, while leaving some ambiguity in how the operators $\hat c_k^\dagger\,,\,\hat c_k$ should be chosen. Let us \eg for any particular mode $k$ use a basis where
    \begin{align}
       \hat Z_k\, &=\, \begin{pmatrix}
           1 & 0 \cr 0 & \textrm{-}1
       \end{pmatrix}
       \otimes \mathbb{I}_r\nonumber \\
       &=\, \begin{pmatrix}
       1 & 0 & \dots & 0 & 0 & \dots \cr 0 & 1 & \dots & 0 & 0 & \dots \cr \vdots & \vdots & \ddots & \vdots & \vdots & \ddots \cr 0 & 0  & \dots & \textrm{-}1 & 0 & \dots \cr 0 & 0  & \dots & 0 & \textrm{-}1 & \dots \cr 0 & 0  & \ddots & \vdots & \vdots & \ddots
       \end{pmatrix}\ .
    \end{align}
    Then any unitary transformation of the form
    \begin{align}
    \label{eq:dressing_choice}
        U\, &=\, \exp\left\lbrace i\begin{pmatrix}
        0 & 0 \cr 0 & 1
        \end{pmatrix} 
        \otimes \begin{pmatrix}
        \theta_1 & 0 & \dots & 0 \cr 0 & \theta_2 & \dots & 0 \cr \vdots & \vdots & \ddots & \vdots \cr 0 & 0  & \dots & \theta_{d/2}
        \end{pmatrix}
        \right\rbrace \nonumber\\
    \end{align}
    will leave $\hat Z_k$ invariant, but it will change the creation and annihilation operators $\hat c_k^\dagger\,,\,\hat c_k\,$. How we should interpret this ambiguity may depend on the answers to the other questions we raised above. If the operators $\hat c_k^\dagger\,,\,\hat c_k\,$ indeed represent real-space modes (and not harmonic modes like we have assumed here) then the ambiguity would be reminiscent of a local gauge freedom - with the Hamiltonian being independent of gauge choice. If we allowed for the extraction of overlapping degrees-of-freedom, then the unitaries in \eqnref{dressing_choice} may change overlaps between the emergent lattice operators and thus be reminiscent of a choice of dressing.
    
\end{itemize}

\section*{Acknowledgements}
\textcopyright 2026. We we would like to thank Sean Carroll and Rogerio Rosenfeld for helpful input to this manuscript. And we would like to thank Arsalan Adil, Andy Albrecht, Laurenz Kohlbach, Varun Kushwaha, and Paul Schneidewind-Telge for helpful discussions. O.F.\ was supported by a Fraunhofer-Schwarzschild-Fellowship at Universitätssternwarte München (LMU observatory) and by DFG's Excellence Cluster ORIGINS (EXC-2094 – 390783311). N.L. was supported by a research grant (42085) from Villum Fonden. We are grateful for the invaluable work of the teams of the public python packages \verb|NumPy| \cite{NumPy}, \verb|SciPy| \cite{scipy}, \verb|CuPy| \cite{Okuta2017CuPyA} and \verb|Matplotlib| \cite{Matplotlib}. No generative AI was used to generate text for this manuscript. When designing GPU-driven code to solve the constrained optimization problem that yields the black dashed line of \figref{minimum_distribution} we benefited from interaction with GPT-5. We thank LMU for the providing a GPT license and everyone - including especially contributors who volunteer their coding expertise in online forums - that helped to make GPT-5 such an incredibly useful coding tool!

\section*{Data availability} 

\verb|Python| tools to predict interaction-strength bounds as a function of Hamiltonian spectral density are publicly available as part of the \verb|GPUniverse|  package, \url{https://github.com/OliverFHD/GPUniverse}~. Our code for finding minimum-interaction qubits from a given Hamiltonian is based on \url{https://github.com/nicolasloizeau/quantum_mereology}~.

\def\aj{AJ}%
\def\araa{ARA\&A}%
\def\apj{ApJ}%
\def\apjl{ApJ}%
\def\apjs{ApJS}%
\def\ao{Appl.~Opt.}%
\def\apss{Ap\&SS}%
\def\aap{A\&A}%
\def\aapr{A\&A~Rev.}%
\def\aaps{A\&AS}%
\def\azh{AZh}%
\def\baas{BAAS}%
\def\jrasc{JRASC}%
\def\memras{MmRAS}%
\def\mnras{MNRAS}%
\def\pra{Phys.~Rev.~A}%
\def\prb{Phys.~Rev.~B}%
\def\prc{Phys.~Rev.~C}%
\def\prd{Phys.~Rev.~D}%
\def\pre{Phys.~Rev.~E}%
\def\prl{Phys.~Rev.~Lett.}%
\def\pasp{PASP}%
\def\pasj{PASJ}%
\def\qjras{QJRAS}%
\def\skytel{S\&T}%
\def\solphys{Sol.~Phys.}%
\def\sovast{Soviet~Ast.}%
\def\ssr{Space~Sci.~Rev.}%
\def\zap{ZAp}%
\def\nat{Nature}%
\def\iaucirc{IAU~Circ.}%
\def\aplett{Astrophys.~Lett.}%
\def\apspr{Astrophys.~Space~Phys.~Res.}%
\def\bain{Bull.~Astron.~Inst.~Netherlands}%
\def\fcp{Fund.~Cosmic~Phys.}%
\def\gca{Geochim.~Cosmochim.~Acta}%
\def\grl{Geophys.~Res.~Lett.}%
\def\jcap{JCAP}%
\def\jcp{J.~Chem.~Phys.}%
\def\jgr{J.~Geophys.~Res.}%
\def\jqsrt{J.~Quant.~Spec.~Radiat.~Transf.}%
\def\memsai{Mem.~Soc.~Astron.~Italiana}%
\def\nphysa{Nucl.~Phys.~A}%
\def\physrep{Phys.~Rep.}%
\def\physscr{Phys.~Scr}%
\def\planss{Planet.~Space~Sci.}%
\def\procspie{Proc.~SPIE}%

\bibliographystyle{iopart-num}
\bibliography{main}

@ARTICLE{Jacobson2016,
       author = {{Jacobson}, Ted},
        title = "{Entanglement Equilibrium and the Einstein Equation}",
      journal = {\prl},
         year = 2016,
        month = may,
       volume = {116},
       number = {20},
          eid = {201101},
        pages = {201101},
          doi = {10.1103/PhysRevLett.116.201101},
archivePrefix = {arXiv},
       eprint = {1505.04753},
 primaryClass = {gr-qc},
       adsurl = {https://ui.adsabs.harvard.edu/abs/2016PhRvL.116t1101J}
}

@ARTICLE{Cao2017,
       author = {{Cao}, ChunJun and {Carroll}, Sean M. and {Michalakis}, Spyridon},
        title = "{Space from Hilbert space: Recovering geometry from bulk entanglement}",
      journal = {\prd},
         year = 2017,
        month = jan,
       volume = {95},
       number = {2},
          eid = {024031},
        pages = {024031},
          doi = {10.1103/PhysRevD.95.024031},
archivePrefix = {arXiv},
       eprint = {1606.08444},
 primaryClass = {hep-th},
       adsurl = {https://ui.adsabs.harvard.edu/abs/2017PhRvD..95b4031C}
}

@ARTICLE{CaoCarroll2018,
       author = {{Cao}, ChunJun and {Carroll}, Sean M.},
        title = "{Bulk entanglement gravity without a boundary: Towards finding Einstein's equation in Hilbert space}",
      journal = {\prd},
         year = 2018,
        month = apr,
       volume = {97},
       number = {8},
          eid = {086003},
        pages = {086003},
          doi = {10.1103/PhysRevD.97.086003},
archivePrefix = {arXiv},
       eprint = {1712.02803},
 primaryClass = {hep-th},
       adsurl = {https://ui.adsabs.harvard.edu/abs/2018PhRvD..97h6003C}
}

@ARTICLE{Zanardi2004,
       author = {{Zanardi}, Paolo and {Lidar}, Daniel A. and {Lloyd}, Seth},
        title = "{Quantum Tensor Product Structures are Observable Induced}",
      journal = {\prl},
         year = 2004,
        month = feb,
       volume = {92},
       number = {6},
          eid = {060402},
        pages = {060402},
          doi = {10.1103/PhysRevLett.92.060402},
archivePrefix = {arXiv},
       eprint = {quant-ph/0308043},
 primaryClass = {quant-ph},
       adsurl = {https://ui.adsabs.harvard.edu/abs/2004PhRvL..92f0402Z}
}

@ARTICLE{Piazza2010,
       author = {{Piazza}, Federico},
        title = "{Glimmers of a Pre-geometric Perspective}",
      journal = {Foundations of Physics},
         year = 2010,
        month = mar,
       volume = {40},
       number = {3},
        pages = {239-266},
          doi = {10.1007/s10701-009-9387-5},
archivePrefix = {arXiv},
       eprint = {hep-th/0506124},
 primaryClass = {hep-th},
       adsurl = {https://ui.adsabs.harvard.edu/abs/2010FoPh...40..239P}
}

@ARTICLE{Carroll_Singh2020,
       author = {{Carroll}, Sean M. and {Singh}, Ashmeet},
        title = "{Quantum mereology: Factorizing Hilbert space into subsystems with quasiclassical dynamics}",
      journal = {\pra},
         year = 2021,
        month = feb,
       volume = {103},
       number = {2},
          eid = {022213},
        pages = {022213},
          doi = {10.1103/PhysRevA.103.022213},
archivePrefix = {arXiv},
       eprint = {2005.12938},
 primaryClass = {quant-ph},
       adsurl = {https://ui.adsabs.harvard.edu/abs/2021PhRvA.103b2213C}
}

@ARTICLE{Singh2018,
       author = {{Singh}, Ashmeet and {Carroll}, Sean M.},
        title = "{Modeling Position and Momentum in Finite-Dimensional Hilbert Spaces via Generalized Pauli Operators}",
      journal = {arXiv e-prints},
         year = 2018,
        month = jun,
          eid = {arXiv:1806.10134},
        pages = {arXiv:1806.10134},
archivePrefix = {arXiv},
       eprint = {1806.10134},
 primaryClass = {quant-ph},
       adsurl = {https://ui.adsabs.harvard.edu/abs/2018arXiv180610134S}
}

@ARTICLE{Cao2019_essay,
       author = {{Cao}, Chunjun and {Chatwin-Davies}, Aidan and {Singh}, Ashmeet},
        title = "{How low can vacuum energy go when your fields are finite-dimensional?}",
      journal = {International Journal of Modern Physics D},
         year = 2019,
        month = jan,
       volume = {28},
       number = {14},
          eid = {1944006},
        pages = {1944006},
          doi = {10.1142/S0218271819440061},
archivePrefix = {arXiv},
       eprint = {1905.11199},
 primaryClass = {hep-th},
       adsurl = {https://ui.adsabs.harvard.edu/abs/2019IJMPD..2844006C}
}

@ARTICLE{Cotler2019,
       author = {{Cotler}, Jordan S. and {Penington}, Geoffrey R. and {Ranard}, Daniel H.},
        title = "{Locality from the Spectrum}",
      journal = {Communications in Mathematical Physics},
         year = 2019,
        month = jun,
       volume = {368},
       number = {3},
        pages = {1267-1296},
          doi = {10.1007/s00220-019-03376-w},
archivePrefix = {arXiv},
       eprint = {1702.06142},
 primaryClass = {quant-ph},
       adsurl = {https://ui.adsabs.harvard.edu/abs/2019CMaPh.368.1267C}
}

@article{Cao:2023gkw,
    author = "Cao, Chunjun and Chemissany, Wissam and Jahn, Alexander and Zimbor\'as, Zolt\'an",
    title = "{Approximate observables from non-isometric maps: de Sitter tensor networks with overlapping qubits}",
    eprint = "2304.02673",
    archivePrefix = "arXiv",
    primaryClass = "hep-th",
    month = "4",
    year = "2023"
}

@article{donnelly+giddings2016_1,
      author         = "Donnelly, William and Giddings, Steven B.",
      title          = "{Diffeomorphism-invariant observables and their nonlocal
                        algebra}",
      journal        = "Phys. Rev.",
      volume         = "D93",
      year           = "2016",
      number         = "2",
      pages          = "024030",
      doi            = "10.1103/PhysRevD.94.029903, 10.1103/PhysRevD.93.024030",
      note           = "[Erratum: Phys. Rev.D94,no.2,029903(2016)]",
      eprint         = "1507.07921",
      archivePrefix  = "arXiv",
      primaryClass   = "hep-th",
      reportNumber   = "NSF-KITP-15-133",
      SLACcitation   = "%%CITATION = ARXIV:1507.07921;%%"
}

@Article{NumPy,
 title         = {Array programming with {NumPy}},
 author        = {Charles R. Harris and K. Jarrod Millman and St{\'{e}}fan J.
                 van der Walt and Ralf Gommers and Pauli Virtanen and David
                 Cournapeau and Eric Wieser and Julian Taylor and Sebastian
                 Berg and Nathaniel J. Smith and Robert Kern and Matti Picus
                 and Stephan Hoyer and Marten H. van Kerkwijk and Matthew
                 Brett and Allan Haldane and Jaime Fern{\'{a}}ndez del
                 R{\'{i}}o and Mark Wiebe and Pearu Peterson and Pierre
                 G{\'{e}}rard-Marchant and Kevin Sheppard and Tyler Reddy and
                 Warren Weckesser and Hameer Abbasi and Christoph Gohlke and
                 Travis E. Oliphant},
 year          = {2020},
 month         = sep,
 journal       = {Nature},
 volume        = {585},
 number        = {7825},
 pages         = {357--362},
 doi           = {10.1038/s41586-020-2649-2},
 publisher     = {Springer Science and Business Media {LLC}},
 url           = {https://doi.org/10.1038/s41586-020-2649-2}
}

@Article{scipy,
 title         = {Array programming with {scipy}},
 author        = {Charles R. Harris and K. Jarrod Millman and St{\'{e}}fan J.
                 van der Walt and Ralf Gommers and Pauli Virtanen and David
                 Cournapeau and Eric Wieser and Julian Taylor and Sebastian
                 Berg and Nathaniel J. Smith and Robert Kern and Matti Picus
                 and Stephan Hoyer and Marten H. van Kerkwijk and Matthew
                 Brett and Allan Haldane and Jaime Fern{\'{a}}ndez del
                 R{\'{i}}o and Mark Wiebe and Pearu Peterson and Pierre
                 G{\'{e}}rard-Marchant and Kevin Sheppard and Tyler Reddy and
                 Warren Weckesser and Hameer Abbasi and Christoph Gohlke and
                 Travis E. Oliphant},
 year          = {2020},
 month         = sep,
 journal       = {Nature},
 volume        = {585},
 number        = {7825},
 pages         = {357--362},
 doi           = {10.1038/s41586-020-2649-2},
 publisher     = {Springer Science and Business Media {LLC}},
 url           = {https://doi.org/10.1038/s41586-020-2649-2}
}

@inproceedings{Okuta2017CuPyA,
  title={CuPy : A NumPy-Compatible Library for NVIDIA GPU Calculations},
  author={Ryosuke Okuta and Yuya Unno and Daisuke Nishino and Shohei Hido and Crissman},
  year={2017},
  url={https://api.semanticscholar.org/CorpusID:41278748}
}

@Article{Matplotlib,
  Author    = {Hunter, J. D.},
  Title     = {Matplotlib: A 2D graphics environment},
  Journal   = {Computing in Science \& Engineering},
  Volume    = {9},
  Number    = {3},
  Pages     = {90--95},
  publisher = {IEEE COMPUTER SOC},
  doi       = {10.1109/MCSE.2007.55},
  year      = 2007
}

@ARTICLE{Chao2017,
  author =	{Chao, Rui and Reichardt, Ben W. and Sutherland, Chris and Vidick, Thomas},
  title =	{{Overlapping Qubits}},
  booktitle =	{8th Innovations in Theoretical Computer Science Conference (ITCS 2017)},
  pages =	{48:1--48:21},
  series =	{Leibniz International Proceedings in Informatics (LIPIcs)},
  ISBN =	{978-3-95977-029-3},
  ISSN =	{1868-8969},
  year =	{2017},
  volume =	{67},
  editor =	{Papadimitriou, Christos H.},
  publisher =	{Schloss Dagstuhl -- Leibniz-Zentrum f{\"u}r Informatik},
  address =	{Dagstuhl, Germany},
  URL =		{https://drops.dagstuhl.de/entities/document/10.4230/LIPIcs.ITCS.2017.48},
  URN =		{urn:nbn:de:0030-drops-81826},
  doi =		{10.4230/LIPIcs.ITCS.2017.48}
}

@ARTICLE{Friedrich2022,
       author = {{Friedrich}, Oliver and {Singh}, Ashmeet and {Dor{\'e}}, Olivier},
        title = "{Toolkit for scalar fields in universes with finite-dimensional Hilbert space}",
      journal = {Classical and Quantum Gravity},
         year = 2022,
        month = dec,
       volume = {39},
       number = {23},
          eid = {235012},
        pages = {235012},
          doi = {10.1088/1361-6382/ac95f0},
archivePrefix = {arXiv},
       eprint = {2201.08405},
 primaryClass = {gr-qc},
       adsurl = {https://ui.adsabs.harvard.edu/abs/2022CQGra..39w5012F}
}

@ARTICLE{Loizeau2023,
       author = {{Loizeau}, Nicolas and {Morone}, Flaviano and {Sels}, Dries},
        title = "{Unveiling order from chaos by approximate 2-localization of random matrices}",
      journal = {Proceedings of the National Academy of Science},
         year = 2023,
        month = sep,
       volume = {120},
       number = {39},
          eid = {e2308006120},
        pages = {e2308006120},
          doi = {10.1073/pnas.2308006120},
archivePrefix = {arXiv},
       eprint = {2303.02782},
 primaryClass = {quant-ph},
       adsurl = {https://ui.adsabs.harvard.edu/abs/2023PNAS..12008006L}
}

@ARTICLE{AlbrechtIglesias2008,
       author = {{Albrecht}, Andreas and {Iglesias}, Alberto},
        title = "{Clock ambiguity and the emergence of physical laws}",
      journal = {\prd},
         year = 2008,
        month = mar,
       volume = {77},
       number = {6},
          eid = {063506},
        pages = {063506},
          doi = {10.1103/PhysRevD.77.063506},
archivePrefix = {arXiv},
       eprint = {0708.2743},
 primaryClass = {hep-th},
       adsurl = {https://ui.adsabs.harvard.edu/abs/2008PhRvD..77f3506A}
}

@ARTICLE{AlbrechtIglesias2015,
       author = {{Albrecht}, Andreas and {Iglesias}, Alberto},
        title = "{Lorentz symmetric dispersion relation from a random Hamiltonian}",
      journal = {\prd},
         year = 2015,
        month = feb,
       volume = {91},
       number = {4},
          eid = {043529},
        pages = {043529},
          doi = {10.1103/PhysRevD.91.043529},
archivePrefix = {arXiv},
       eprint = {1003.2566},
 primaryClass = {hep-th},
       adsurl = {https://ui.adsabs.harvard.edu/abs/2015PhRvD..91d3529A}
}

@ARTICLE{Adil2024,
       author = {{Adil}, Arsalan and {Rudolph}, Manuel S. and {Arrasmith}, Andrew and {Holmes}, Zo{\"e} and {Albrecht}, Andreas and {Sornborger}, Andrew},
        title = "{A Search for Classical Subsystems in Quantum Worlds}",
      journal = {arXiv e-prints},
         year = 2024,
        month = mar,
          eid = {arXiv:2403.10895},
        pages = {arXiv:2403.10895},
          doi = {10.48550/arXiv.2403.10895},
archivePrefix = {arXiv},
       eprint = {2403.10895},
 primaryClass = {quant-ph},
       adsurl = {https://ui.adsabs.harvard.edu/abs/2024arXiv240310895A}
}

@ARTICLE{Zanardi2024,
       author = {{Zanardi}, Paolo and {Dallas}, Emanuel and {Andreadakis}, Faidon and {Lloyd}, Seth},
        title = "{Operational Quantum Mereology and Minimal Scrambling}",
      journal = {Quantum},
         year = 2024,
        month = jul,
       volume = {8},
        pages = {1406},
          doi = {10.22331/q-2024-07-11-1406},
archivePrefix = {arXiv},
       eprint = {2212.14340},
 primaryClass = {quant-ph},
       adsurl = {https://ui.adsabs.harvard.edu/abs/2024Quant...8.1406Z}
}

@ARTICLE{Friedrich2024,
       author = {{Friedrich}, Oliver and {Cao}, ChunJun and {Carroll}, Sean M. and {Cheng}, Gong and {Singh}, Ashmeet},
        title = "{Holographic phenomenology via overlapping degrees of freedom}",
      journal = {Classical and Quantum Gravity},
         year = 2024,
        month = oct,
       volume = {41},
       number = {19},
          eid = {195003},
        pages = {195003},
          doi = {10.1088/1361-6382/ad6e4d},
archivePrefix = {arXiv},
       eprint = {2402.11016},
 primaryClass = {hep-th},
       adsurl = {https://ui.adsabs.harvard.edu/abs/2024CQGra..41s5003F}
}

@Article{Loizeau2024,
author={Loizeau, Nicolas
and Sels, Dries},
title={Quantum Mereology and Subsystems from the Spectrum},
journal={Foundations of Physics},
year={2024},
month={Dec},
day={21},
volume={55},
number={1},
pages={3},
issn={1572-9516},
doi={10.1007/s10701-024-00813-2},
url={https://doi.org/10.1007/s10701-024-00813-2}
}

@ARTICLE{Kabernik2020,
       author = {{Kabernik}, Oleg and {Pollack}, Jason and {Singh}, Ashmeet},
        title = "{Quantum state reduction: Generalized bipartitions from algebras of observables}",
      journal = {\pra},
         year = 2020,
        month = mar,
       volume = {101},
       number = {3},
          eid = {032303},
        pages = {032303},
          doi = {10.1103/PhysRevA.101.032303},
archivePrefix = {arXiv},
       eprint = {1909.12851},
 primaryClass = {quant-ph},
       adsurl = {https://ui.adsabs.harvard.edu/abs/2020PhRvA.101c2303K}
}

@ARTICLE{Pollack_Singh2019,
       author = {{Pollack}, Jason and {Singh}, Ashmeet},
        title = "{Towards space from Hilbert space: finding lattice structure in finite-dimensional quantum systems}",
      journal = {Quantum Studies: Mathematics and Foundations},
         year = 2019,
        month = jun,
       volume = {6},
       number = {2},
        pages = {181-200},
          doi = {10.1007/s40509-018-0176-8},
archivePrefix = {arXiv},
       eprint = {1801.10168},
 primaryClass = {quant-ph},
       adsurl = {https://ui.adsabs.harvard.edu/abs/2019QSMF....6..181P}
}

@misc{soulas2025a,
      title={On the emergence of preferred structures in quantum theory}, 
      author={Antoine Soulas and Guilherme Franzmann and Andrea Di Biagio},
      year={2025},
      eprint={2512.07468},
      archivePrefix={arXiv},
      primaryClass={quant-ph},
      url={https://arxiv.org/abs/2512.07468}, 
}

@misc{vanrietvelde2025,
      title={Partitions in quantum theory}, 
      author={Augustin Vanrietvelde and Octave Mestoudjian and Pablo Arrighi},
      year={2025},
      eprint={2506.22218},
      archivePrefix={arXiv},
      primaryClass={quant-ph},
      url={https://arxiv.org/abs/2506.22218}, 
}

@misc{zanardi2025,
      title={Mereological Quantum Phase Transitions}, 
      author={Paolo Zanardi and Emanuel Dallas and Faidon Andreadakis},
      year={2025},
      eprint={2510.06389},
      archivePrefix={arXiv},
      primaryClass={quant-ph},
      url={https://arxiv.org/abs/2510.06389}, 
}

@article{loizeau2025,
  title = {Opening Krylov Space to Access All-Time Dynamics via Dynamical Symmetries},
  author = {Loizeau, Nicolas and Bu\ifmmode \check{c}\else \v{c}\fi{}a, Berislav and Sels, Dries},
  journal = {Phys. Rev. Lett.},
  volume = {135},
  issue = {20},
  pages = {200401},
  numpages = {9},
  year = {2025},
  month = {Nov},
  publisher = {American Physical Society},
  doi = {10.1103/qpyp-1mfj},
  url = {https://link.aps.org/doi/10.1103/qpyp-1mfj}
}

@article{buca2023,
  title = {Unified Theory of Local Quantum Many-Body Dynamics: Eigenoperator Thermalization Theorems},
  author = {Bu\ifmmode \check{c}\else \v{c}\fi{}a, Berislav},
  journal = {Phys. Rev. X},
  volume = {13},
  issue = {3},
  pages = {031013},
  numpages = {30},
  year = {2023},
  month = {Aug},
  publisher = {American Physical Society},
  doi = {10.1103/PhysRevX.13.031013},
  url = {https://link.aps.org/doi/10.1103/PhysRevX.13.031013}
}

@article{Mori2018,
doi = {10.1088/1361-6455/aabcdf},
url = {https://doi.org/10.1088/1361-6455/aabcdf},
year = {2018},
month = {may},
publisher = {IOP Publishing},
volume = {51},
number = {11},
pages = {112001},
author = {Mori, Takashi and Ikeda, Tatsuhiko N and Kaminishi, Eriko and Ueda, Masahito},
title = {Thermalization and prethermalization in isolated quantum systems: a theoretical overview},
journal = {Journal of Physics B: Atomic, Molecular and Optical Physics},
}

@article{Polkovnikov2011,
  title = {Colloquium: Nonequilibrium dynamics of closed interacting quantum systems},
  author = {Polkovnikov, Anatoli and Sengupta, Krishnendu and Silva, Alessandro and Vengalattore, Mukund},
  journal = {Rev. Mod. Phys.},
  volume = {83},
  issue = {3},
  pages = {863--883},
  numpages = {0},
  year = {2011},
  month = {Aug},
  publisher = {American Physical Society},
  doi = {10.1103/RevModPhys.83.863},
  url = {https://link.aps.org/doi/10.1103/RevModPhys.83.863}
}

@Article{Rigol2008,
author={Rigol, Marcos
and Dunjko, Vanja
and Olshanii, Maxim},
title={Thermalization and its mechanism for generic isolated quantum systems},
journal={Nature},
year={2008},
month={Apr},
day={01},
volume={452},
number={7189},
pages={854-858},
issn={1476-4687},
doi={10.1038/nature06838},
url={https://doi.org/10.1038/nature06838}
}

@article{Sala2020,
  title = {Ergodicity Breaking Arising from Hilbert Space Fragmentation in Dipole-Conserving Hamiltonians},
  author = {Sala, Pablo and Rakovszky, Tibor and Verresen, Ruben and Knap, Michael and Pollmann, Frank},
  journal = {Phys. Rev. X},
  volume = {10},
  issue = {1},
  pages = {011047},
  numpages = {19},
  year = {2020},
  month = {Feb},
  publisher = {American Physical Society},
  doi = {10.1103/PhysRevX.10.011047},
  url = {https://link.aps.org/doi/10.1103/PhysRevX.10.011047}
}

@article{Wigner2023,
 ISSN = {0003486X},
 URL = {http://www.jstor.org/stable/1970079},
 author = {Eugene P. Wigner},
 journal = {Annals of Mathematics},
 number = {3},
 pages = {548--564},
 publisher = {Annals of Mathematics},
 title = {Characteristic Vectors of Bordered Matrices With Infinite Dimensions},
 urldate = {2023-03-02},
 volume = {62},
 year = {1955}
}

@article{Guhr1998,
   author = {Thomas Guhr and Axel Müller–Groeling and Hans A. Weidenmüller},
   doi = {10.1016/S0370-1573(97)00088-4},
   issn = {03701573},
   issue = {4-6},
   journal = {Physics Reports},
   month = {6},
   pages = {189-425},
   title = {Random-matrix theories in quantum physics: common concepts},
   volume = {299},
   year = {1998},
}

@article{Atas2013,
   author = {Y Y Atas and E Bogomolny and O Giraud and G Roux},
   doi = {10.1103/PhysRevLett.110.084101},
   issue = {8},
   journal = {Phys. Rev. Lett.},
   month = {2},
   pages = {84101},
   publisher = {American Physical Society},
   title = {Distribution of the Ratio of Consecutive Level Spacings in Random Matrix Ensembles},
   volume = {110},
   url = {https://link.aps.org/doi/10.1103/PhysRevLett.110.084101},
   year = {2013},
}

@article{Lashkari2015,
    author = "Lashkari, Nima and Van Raamsdonk, Mark",
    title = "{Canonical Energy is Quantum Fisher Information}",
    eprint = "1508.00897",
    archivePrefix = "arXiv",
    primaryClass = "hep-th",
    doi = "10.1007/JHEP04(2016)153",
    journal = "JHEP",
    volume = "04",
    pages = "153",
    year = "2016"
}

\appendix
\addtocontents{toc}{\fixappendix}

\section{GUE conventions \& Gaussian Pauli-string coefficients}
\label{app:GUE}

Let a Hamiltonian $\hat H$ be drawn from the Gaussian Unitary Ensemble (i.e.\ let it be a \enquote{Wigner matrix}). This means that the entries in the upper triangle of $\hat H$ are all independent Gaussian numbers with
\begin{equation}
    \langle H_{ij} \rangle = 0\ \ \ \ \&\ \ \ \langle |H_{ij}|^2 \rangle = \frac{1}{d}\ ,
\end{equation}
where $d = \mathrm{dim} \mathcal{H}$ is the dimension of the Hilbert space. In the limit of $d\rightarrow \infty$ the empirical eigenvalue distribution of such matrices tends towards the Wigner semicircle distribution,
\begin{equation}
    p_{\hat H}(\lambda) \approx \frac{1}{2\pi}\sqrt{4 - \lambda^2}\ .
\end{equation}

When we view the space of Hermitian matrices as a real Hilbert space endowed with the Hilbert-Schmidt product $\langle A, B \rangle = \mathrm{Tr}(AB)$, then the GUE is actually just the isotropic Gaussian distribution in that space with variance $\sigma^2 = 1/d$ in each dimension. To see this, consider the basis matrices $\mathcal{B} = \lbrace\hat D_{a}\,,\,\hat R_{ab}\,,\,\hat I_{ab}\rbrace$ given by
\begin{align}
    (D_{a})_{ij} &= \delta_{ia} \delta_{ja} \nonumber \\
    (R_{ab})_{ij} &= \frac{1}{\sqrt{2}}\left(\delta_{ia} \delta_{jb} + \delta_{ja} \delta_{ib}\right) \nonumber \\
    (I_{ab})_{ij} &= \frac{1}{\sqrt{2}}\left(\mathrm{i}\delta_{ia} \delta_{jb} - \mathrm{i}\delta_{ja} \delta_{ib}\right)\ .
\end{align}
These matrices are orthonormal \wrt the Hilbert-Schmidt product and they form a complete basis of the Hermitian matrices, \ie any Hamiltonian $\hat H$ can be decomposed as
\begin{equation}
    \hat H = \sum_{a} h^{D,a} \hat D_{a}\ +\ \sum_{a<b} h^{R,ab} \hat R_{ab}\ +\ \sum_{a<b} h^{I,ab} \hat I_{ab}\ .
\end{equation}
Now choosing all the coefficients $\lbrace h^{D,a} \,,\, h^{R,ab}\,,\, h^{I,ab}\rbrace$ from independent Gaussian distributions with the same variance $\sigma^2 = 1/d$ will exactly lead to $\langle |H_{ij}|^2 \rangle = \frac{1}{d}$, as we required for the GUE.

But note that because the GUE is just an isotropic Gaussian in the space of all Hermitian $\hat H$ (where \enquote{isotropic} means that the covariance matrix of the Gaussian is proportional to the identity) we may rotate to any other orthonormal basis of that space and draw independent components along the axes of this new basis, and we will still obtain draws from that same original distribution. For that reason, the GUE is equivalent to drawing random Gaussian coefficients in a decomposition of $\hat H$ in terms of Pauli strings - as we had claimed in \secref{app_to_GUE}.

\section{Hamilton operator at candidate dofs}
\label{app:structure_of_Hint}

In this appendix we demonstrate that at the interaction minima, the (traceless part of the) Hamiltonian of our system can always be written as
\begin{equation}
    \hat H\, =\, \frac{E_q}{2}\,\sigma_z \otimes \mathbb{I}_r\, +\, \mathbb{I}_q \otimes \hat H_{r}\, +\, \sigma_z \otimes \hat X_{\mathrm{int}}\ ,
\end{equation}
where the interaction part of the Hamiltonian, $\sigma_z \otimes \hat X_{\mathrm{int}}$, commutes with the self-Hamiltonians. 

We start by considering any factorization $\mathcal{H} \simeq \mathcal{H}_q \otimes \mathcal{H}_r$. In any such factorization the Hamiltonian can always be written as
\begin{align}
    \hat H\, &=\, \frac{E_q}{2}\,\sigma_z \otimes \mathbb{I}_r\, +\, \mathbb{I}_q \otimes \hat H_{r}\, +\, \hat H_{\mathrm{int}}\\
    &\equiv\, \hat H_0 \, +\, \hat H_{\mathrm{int}}\ ,
\end{align}
where $\mathrm{Tr}_q \hat H_{\mathrm{int}} = 0$ and $\mathrm{Tr}_r \hat H_{\mathrm{int}} = 0\,$, and where in the 2nd line we have defined $\hat H_0$ to be the non-interacting part of the Hamiltonian. Now let $\lbrace \ket{i}\rbrace$ be an eigenbasis of $\hat H_0$ and let us assume that $\hat H_{\mathrm{int}}$ is \emph{not} diagonal in that basis. We can then consider an infinitesimal Schrieffer–Wolff-like transformation
\begin{equation}
    \hat H \rightarrow \hat H_\alpha = e^{\alpha \hat S} \hat H e^{-\alpha \hat S}
\end{equation}
where the elements of $\hat S$ in the basis $\lbrace \ket{i}\rbrace$ are given by
\begin{equation}
    S_{ij} = \left\lbrace \begin{matrix}
    \frac{H_{\mathrm{int},ij}}{d_i - d_j}\ \ &\mathrm{for}\ d_i\neq d_j \cr 0\ \ &\mathrm{else} 
    \end{matrix}\right.\ .
\end{equation}
Note that for $d_i \neq d_j$
\begin{align}
    [\hat S , \hat H_0]_{ij} &= \frac{H_{\mathrm{int},ij}\, d_j}{d_i - d_j} - \frac{H_{\mathrm{int},ij}\, d_i}{d_i - d_j}\nonumber \\
    &= -H_{\mathrm{int},ij} \\
    \Rightarrow [\hat S , \hat H_0] &= - \hat V_{\mathrm{int}}\ ,
\end{align}
where the matrix elements of $\hat V_{\mathrm{int}}$ are the same as those of $\hat H_{\mathrm{int}}$ except when $d_i = d_j$, where we set them to $0\,$.

The transformation $\hat H \rightarrow \hat H_\alpha'$ can be interpreted as a change of Hilbert space factorization (+ internal transformations of the factors), and we want to show that -- under the assumption that $\hat S \neq 0$, i.e.\ that $\hat H_{\mathrm{int}}$ is not diagonal in the eigenbasis of $\hat H_0$ -- the derivative of $\mathrm{Tr}(\hat H_{\mathrm{int}}^2)$ is always negative. If this were true, then any factorization where $\hat S \neq 0$ cannot be a local minimum of our loss function.

To execute this plan, let us first note that
\begin{equation}
    \frac{\dd }{\dd \alpha} \mathrm{Tr}(\hat H_{\mathrm{int},\alpha}^2) = 2 \mathrm{Tr}\left(\hat H_{\mathrm{int}} \frac{\dd }{\dd \alpha}\hat H_{\mathrm{int},\alpha}\right)\ .
\end{equation}
Note also that $\hat H_{\mathrm{int}}$ can be computed from $\hat H$ as
\begin{equation}
    \hat H_{\mathrm{int}} = \hat H - \frac{1}{d/2}\mathrm{Tr}_r(\hat H) \otimes \mathbb{I}_r\, -\, \mathbb{I}_q \otimes \frac{1}{2}\mathrm{Tr}_q(\hat H)\ .
\end{equation}
Consequently, we can express the derivative of $\hat H_{\mathrm{int}}$ \wrt $\alpha$ as
\begin{align}
    \frac{\dd \hat H_{\mathrm{int},\alpha}}{\dd \alpha} &= \frac{\dd \hat H_\alpha}{\dd \alpha} - \frac{1}{d/2}\mathrm{Tr}_r\left(\frac{\dd \hat H_\alpha}{\dd \alpha}\right) \otimes \mathbb{I}_r\nonumber \\
    &\ \ -\, \mathbb{I}_q \otimes \frac{1}{2}\mathrm{Tr}_q\left(\frac{\dd \hat H_\alpha}{\dd \alpha}\right)\ .
\end{align}
Since all partial traces of $\hat H_{\mathrm{int}}$ vanish, this means that
\begin{equation}
    \mathrm{Tr}\left(\hat H_{\mathrm{int}} \frac{\dd }{\dd \alpha}\hat H_{\mathrm{int},\alpha}\right) = \mathrm{Tr}\left(\hat H_{\mathrm{int}} \frac{\dd }{\dd \alpha}\hat H_{\alpha}\right)\ .
\end{equation}
Thus, the derivative of our loss function \wrt $\alpha$ is given by
\begin{align}
    &\left. \frac{\dd }{\dd \alpha}\mathrm{Tr}\left(\hat H_{\mathrm{int},\alpha}^2\right)\right|_{\alpha = 0}\nonumber \\
    &= 2 \mathrm{Tr}\left( \hat H_{\mathrm{int}} \left. \frac{\dd }{\dd \alpha} \hat H_{\alpha}\right|_{\alpha=0}\right)\nonumber \\
    &= 2 \mathrm{Tr}\left( \hat H_{\mathrm{int}}\ [\hat S, \hat H]\right)\nonumber \\
    &= 2 \mathrm{Tr}\left( \hat H_{\mathrm{int}}\ [\hat S, \hat H_0]\right) + 2 \mathrm{Tr}\left( \hat H_{\mathrm{int}}\ [\hat S, \hat H_{\mathrm{int}}]\right)\nonumber \\
    &= -2 \mathrm{Tr}\left( \hat H_{\mathrm{int}}\, \hat V_{\mathrm{int}}\right)\nonumber \\
    &\ \ \ + 2 \mathrm{Tr}\left( \hat H_{\mathrm{int}}\hat S \hat H_{\mathrm{int}}\right) - 2 \mathrm{Tr}\left( \hat H_{\mathrm{int}}\hat H_{\mathrm{int}} \hat S \right)\nonumber \\
    &= -2 \mathrm{Tr}\left( \hat H_{\mathrm{int}}\, \hat V_{\mathrm{int}}\right)\nonumber \\
    &= -2 \mathrm{Tr}\left(\hat V_{\mathrm{int}}^2\right)\ <\ 0\ ,
\end{align}
where in the last step we have used the fact that $\hat V_{\mathrm{int}}$ has no diagonal elements in the basis $\lbrace \ket{i}\rbrace\,$. So we have proven: if $\hat V_{\mathrm{int}}$ is non-zero, i.e.\ if $\hat H_{\mathrm{int}}$ is not diagonal in the eigenbasis of $H_0$, then an infinitesimal unitary transformation can change the Hilbert space factorization to a factorization with lower interaction strength. Consequently, at local minima of interaction strength, $\hat H_{\mathrm{int}}$ must be diagonal in the eigenbasis of $H_0$. Since by convention we wrote the eigen-Hamiltonian of the qubit factor as $\propto \sigma_z\otimes \mathbb{I}_r$, this means that at those minima $\hat H_{\mathrm{int}}$ is of the form
\begin{equation}
    \hat H_{\mathrm{int}} = \sigma_z \otimes \hat X_{\mathrm{int}}\ .
\end{equation}

\section{Discrete Pauli field}
\label{app:Dirac}

\begin{figure*}
\centering
  \includegraphics[width=0.8\textwidth]{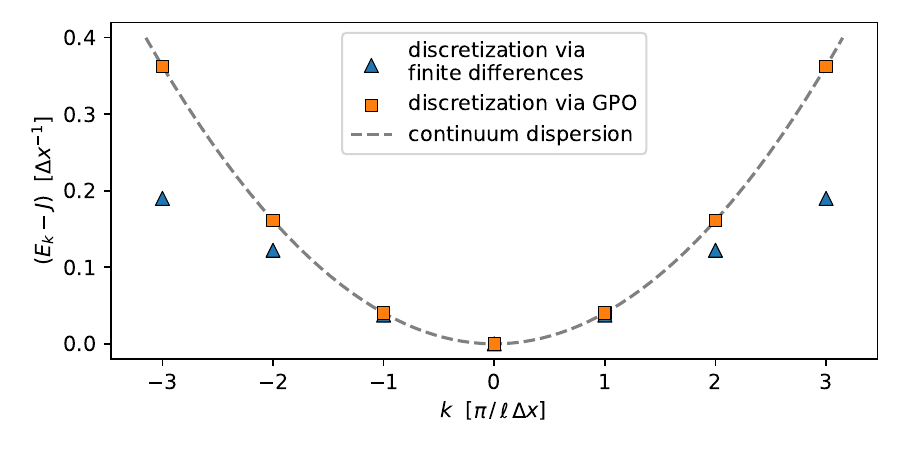}
   \caption{Dispersion relation of the discretized Pauli field when using different schemes to represent the Laplacian $\Delta$ on the grid.}
   \label{fi:dispersions}
\end{figure*}

In the non-relativistic limit, the Dirac field in 1+1 dimensions decouples into 4 copies of a field with the Hamiltonian
\begin{equation}
    \hat H = \int \dd x\ \hat b^\dagger(x) \left(m - \frac{1}{2m}\frac{\partial^2}{\partial x^2} \right) \hat b(x)
\end{equation}
where $\hat b^\dagger(x)$ creates a particle at location $x$ and satisfies
\begin{equation}
    \lbrace b^\dagger(x)\, ,\, b(y)\rbrace = \delta(x-y)\ .
\end{equation}
We will refer to this field as \enquote{Pauli field}. Discretizing the Pauli field on a lattice with grid spacing $\Delta x$ and using a finite-differences scheme to approximate $\partial^2/\partial x^2$ would yield
\begin{align}
    &\hat H = \nonumber \\
    &\sum_i \Delta x\ \hat b^\dagger(x_i) \left(m\, \hat b(x_i) - \frac{\hat b(x_{i-1}) + \hat b(x_{i+1}) - 2\hat b(x_i)}{2m\Delta x^2} \right) \ .
\end{align}
And defining
\begin{equation}
    \hat B_i \equiv \sqrt{\Delta x}\, \hat b(x_i)\ \Rightarrow\ \lbrace \hat B_i^\dagger\, ,\, \hat B_j\rbrace = \delta_{ij}
\end{equation}
this becomes
\begin{align}
    &\hat H = \sum_i  \hat B_i^\dagger \left(m \hat B_i - \frac{\hat B_{i+1} + \hat B_{i-1} - 2\hat B_{i}}{2m\Delta x^2} \right)\nonumber \\
    &= \sum_i \left[ \left(m + \frac{1}{m\, \Delta x^2}\right)\hat B_i^\dagger \hat B_i - \frac{\hat B_i^\dagger\hat B_{i+1} + \hat B_{i+1}^\dagger\hat B_{i}}{2m\Delta x^2} \right]\ .
\end{align}
Note that on small grids (like the ones we have used in our numerical results) the finite-differences-discretization employed above yields dispersion relations of the free field's Fourier modes that strongly deviate from the continuum limit - \cf the blue triangles and grey, dashed line in \figref{dispersions}. To ease this problem we choose to represent the derivative $\hat P^2 = -\partial^2/\partial x^2$ with the help of Generalized Pauli Operators (GPOs) instead of finite-differences. GPO can be seen as the most faithful (but still approximate) implementation of the Heisenberg commutator in finite dimensions  \cite{Singh2018, Friedrich2022}. The resulting conjugate operator pair $\hat P, \hat X$ induce shifts in each others spectra with periodic boundary conditions and their spectra consist of linearly spaced eigenvalues which align exactly with the position space and Fourier space grid of the discretization. As a consequence, the eigenvalues of $\hat P^2$ align much closer with the continuum limit, as we demonstrate with the orange squares in \figref{dispersions}. The (field-)Hamiltonian in this alternative discretization scheme becomes
\begin{align}
    \Delta x\hat H_{\mathrm{no-tr.}} &= \sum_{i,j} M_{ij}\, \left( \hat{B}_i^\dagger \hat{B}_j - \frac{1}{2}\delta_{ij}\right) \ ,
\end{align}
where the coefficients $M_{ij}$ are given by the matrix elements of $\hat P^2/2J$ in the GPO-scheme (see \eg \cite{Singh2018} for details), and where we have explicitly removed the trace of $\hat H\,$.

\section{Information minima \& snap-shot mereology}
\label{app:Wishart}

\begin{figure*}
\centering
  \includegraphics[width=0.9\textwidth]{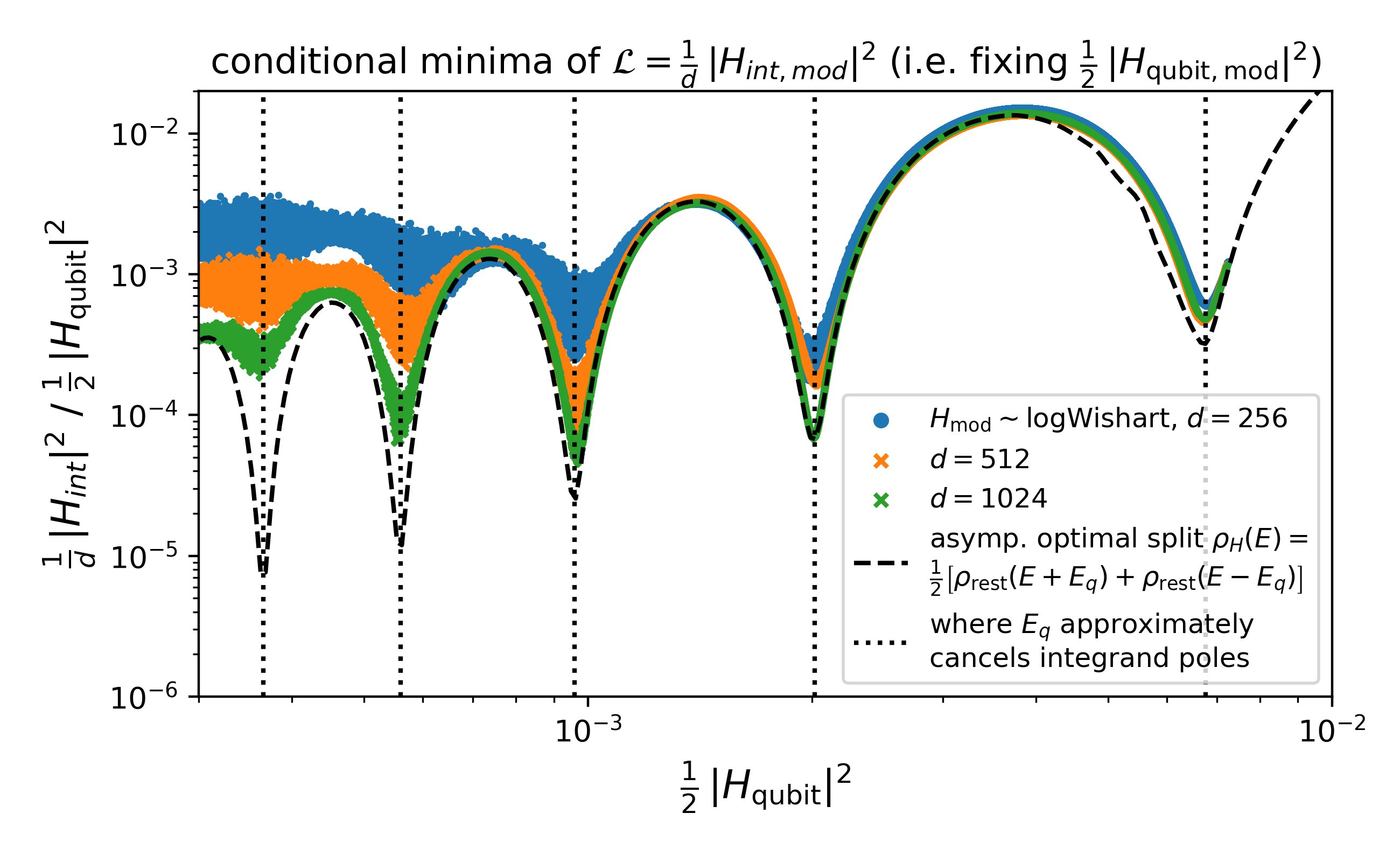}
   \caption{Same as \figref{minimum_distribution} but for Hamiltonians drawn from the log-Wishart ensemble. These can be interpreted as modular Hamiltonians corresponding to standard random density matrices.}
   \label{fi:conditional_minima_logWishart}
\end{figure*}

In this appendix we transfer the results of \secref{method} to a different situation that also highlights a different point-of-view on quantum mereology itself. We consider matrices of the form
\begin{equation}
    \hat H_{\mathrm{mod}} = -\ln\hat W\ ,
\end{equation}
where $\hat W$ is drawn from a complex Wishart distribution with $n$ degrees-of-freedom (not to be confused with the emergent \enquote{degrees-of-freedom} we are investigating in the main text). In the limit $d\rightarrow \infty$ and $\frac{d}{n}\rightarrow \gamma = \mathrm{const.}$ such $\hat H_{\mathrm{mod}}$ are distributed approximately (and up to a constant energy offset) like modular Hamiltonians of random density matrices that are obtained by considering Haar-random pure states in a Hilbert space of dimension $d\times n$ and then tracing out a random $n$-dimensional factor from that Hilbert space. 

When factoring out a qubit from the remaining, $d$-dimensional Hilbert space, $\hat H_{\mathrm{mod}}$ can be again decomposed into self- and interaction-Hamiltonians. But our loss function $\mathrm{Tr}\left(\hat H_{\mathrm{int}}^2\right)$ now acquires an information-theoretic meaning: it can be shown that the mutual information between $\mathcal{H}_q$ and $\mathcal{H}_r$ is at leading order in the terms of $\hat H_{\mathrm{mod}}$ given by \cite{Lashkari2015}
\begin{align}
\label{eq:Hint_vs_I}
    \mathcal{I}(q:r\,|\,\hat W) \approx \frac{1}{2d}\mathrm{Tr}\left(\hat H_{\mathrm{int}}^2\right)\ .
\end{align}
This leading order expression will be a good approximation for $d \ll n$, or equivalently, for $\hat W \approx \mathbb{I} - \hat H_{\mathrm{mod}}\,$. So in that limit we are identifying qubits that have minimum mutual information with their complement.

In \figref{conditional_minima_logWishart} we show results similar to those of \figref{minimum_distribution} but in this modified situation. The different symbols show interaction minima ($\sim$ mutual information minima) as a function of the norm of the qubits modular self-Hamiltonian. Again, a discrete series of peaks with particularly low loss function appear. To predict the asymptotic location of these minima we simply need to replace the asymptotic eigenvalue distribution of the (modular) Hamiltonian in our calculations of \secref{GUE_species}. There we had considered GUE Hamiltonians whose eigenvalue density asymptotes to the semi-circle law. Now, our modular Hamiltonians are drawn from a log-Wishart distribution which follows a log-Marchenko-Pastur law, i.e.
\begin{align}
    &\rho_{\mathrm{total}}(\lambda) = \nonumber \\
    &\frac{\sqrt{\left(e^{-\lambda} - \left[1-\sqrt{\gamma}\right]^2\right)\left(\left[1+\sqrt{\gamma}\right]^2 - e^{-\lambda}\right)}}{2\pi\gamma}\ .
\end{align}
Replacing the semi-circle law with the log-Marchenko-Pastur law yields the dashed line in \figref{conditional_minima_logWishart} which indeed seems to represent a lower bound for our  empirically found minima. And in the limit of large dimension $d$, the minima seem to approach this bound. Note that this time the characteristic function $\phi_{\mathrm{total}}$ corresponding to $\rho_{\mathrm{total}}$ is complex and has no exact zeros that can cancel integrand poles as in \secref{method}. But minima of the absolute value $|\phi_{\mathrm{total}}|$ can still dampen those poles. And the location of these minima still correctly predict the location of low-interaction peaks - as we demonstrate with the vertical dotted lines in \figref{conditional_minima_logWishart}.

Interpreting the degrees-of-freedom identified by our methodology as minimum-information factors as opposed to minimum interaction factors may be viewed as a \enquote{snap-shot} notion of mereology where preferred Hilbert space factors are not selected by dynamical considerations but by a static, information-theoretical criterion applied to a given density matrix. Such a point-of-view may be required in the light of work such as \cite{Cao2017, CaoCarroll2018}, where the emergence of classical spacetime is dependent on the structure of quantum information between Hilbert space factors. It can also loosely be motivated by analogy with large-scale structure cosmology: cosmologists do not observe the dynamics of the universe but only one snap-shot along our backward light-cone. The information extracted from that light-cone (which is essentially a density matrix in the Hilbert space of the inflaton field) is then arranged into cosmologically useful degrees-of-freedom and thus into a coherent story of inflation and structure formation. In that sense, cosmology itself is an implementation of (snapshot-)mereology. We will investigate such a point-of-view on mereology further in follow-up work.

\end{document}